\documentclass[10pt,conference,letterpaper]{IEEEtran}

\IEEEoverridecommandlockouts

\usepackage[T1]{fontenc} %%%%% Optional T1 Font Encoding %%%%%

\usepackage[all]{nowidow}
\usepackage{mathrsfs}
\usepackage{boxedminipage}
\usepackage[numbers,sort&compress]{natbib}

\usepackage[acronyms,nonumberlist,nopostdot,nomain,nogroupskip]{glossaries}
\usepackage{balance}
\usepackage{algorithmic}
\usepackage{nccmath}
\usepackage{caption}
\usepackage[percent]{overpic}
\usepackage{amsmath,amssymb,latexsym, amsfonts,mathrsfs,amsthm}
\usepackage{hhline}
\usepackage{graphicx}
\usepackage{textcomp}
\usepackage{xcolor}
\usepackage{enumitem}
\usepackage{dirtytalk}
\newtheorem{theorem}{Theorem}[section]
\newtheorem{lemma}[theorem]{Lemma}
\newtheorem{definition}{Definition}[section]

\usepackage{mathtools}
\usepackage{graphicx}
\usepackage{comment}
\usepackage{enumitem}
\usepackage{multirow}
\usepackage{subfigure,epsfig,epstopdf}
\usepackage[ruled,noend, vlined,linesnumbered]{algorithm2e}
\usepackage{url}
\usepackage{soul}
\usepackage{epsfig}
\usepackage{epstopdf}
\usepackage{newfloat}
\usepackage{setspace}
\usepackage{graphics}
\usepackage{makecell}

\DeclareFloatingEnvironment[
    fileext=los,
    listname={List of Problems},
    name=Problem,
    placement=htbp,
]{problem}

\newcommand{\FW}{\texttt{A$^2$-UAV}\xspace}
\newcommand{\greedyprob}{\textsc{Greedy}-\gls{prob}\xspace}
\newcommand{\optprob}{\textsc{Opt}-\gls{prob}\xspace}

\newcommand{\U}{\gls{uav}\xspace}
\newcommand{\Us}{\glspl{uav}\xspace}

\begin{document}

\title{\FW: Application-Aware Content and Network Optimization of Edge-Assisted UAV Systems 
}
\author{\IEEEauthorblockN{Andrea Coletta$^\dagger$, Flavio Giorgi$^\dagger$, Gaia Maselli$^\dagger$, Matteo Prata$^\dagger$, \\ Domenicomichele Silvestri$^\dagger$, Jonathan Ashdown$^\ddagger$, and Francesco Restuccia$^*$}\\\vspace{-0.3cm}
\IEEEauthorblockA{
$^\dagger$Department of Computer Science, Sapienza University of Rome, Italy\\
$^\ddagger$Air Force Research Laboratory, United States\\
$^*$Institute for the Wireless Internet of Things, Northeastern University, United States\\
}
\normalsize Corresponding Author: Andrea Coletta, e-mail: coletta@di.uniroma1.it\vspace{-.5cm}
\thanks{Approved for Public Release; Distribution Unlimited: AFRL-2022-1069.
}
}

\maketitle

\newacronym{mas}{MAS}{Mobile autonomous system}
\newacronym{cv}{CV}{computer vision}
\newacronym{dnn}{DNN}{deep neural network}
\newacronym{dl}{DL}{deep learning}
\newacronym{uav}{UAV}{Unmanned Aerial Vehicle}
\newacronym{ai}{AI}{Artificial Intelligence}
\newacronym{milp}{MILP}{Mixed Integer Linear Programming}
\newacronym{ugv}{UGV}{Unmanned Ground Vehicle}
\newacronym{stba}{STBA}{Steiner-Tree-Based Algorithm}
\newacronym{aoi}{AoI}{Area of Interest}
\newacronym{cnn}{CNN}{Convolutional Neural Network}
\newacronym{fps}{fps}{frames per second}
\newacronym{gpu}{GPU}{Graphics Processing Unit}
\newacronym{prob}{\texttt{A$^2$-TPP}}{Application-Aware Task Planning Problem}
\newacronym{dct}{DCT}{Discrete Cosine Transform}
\newacronym{sta}{\texttt{A$^2$-TA}}{Application-Aware Task Analyzer}

\begin{abstract}
To perform advanced surveillance, \Us require the execution of edge-assisted \gls{cv} tasks. In multi-hop UAV networks, the successful transmission of these tasks to the edge is severely challenged due to severe bandwidth constraints. For this reason, we propose a novel \FW framework to optimize the number of correctly executed tasks at the edge. In stark contrast with existing art, we take an \textit{application-aware} approach  and formulate a  novel \gls{prob} that takes into account (i) the relationship between \gls{dnn} accuracy and image compression for the classes of interest based on the available dataset, (ii) the target positions, (iii) the current energy/position of the \Us to optimize routing, data pre-processing and target assignment for each \U. We demonstrate \gls{prob} is NP-Hard and propose a polynomial-time algorithm to solve it efficiently. We extensively evaluate \FW through real-world experiments with a testbed composed by four DJI Mavic Air 2 \Us. We consider state-of-the-art image classification tasks with four different DNN models (i.e., DenseNet, ResNet152, ResNet50 and MobileNet-V2) and object detection tasks using YoloV4 trained on the ImageNet dataset.  Results show that \FW attains on average around $38\%$  more accomplished tasks than the state of the art, with $400\%$ more accomplished tasks when the number of targets increase significantly. To allow full reproducibility, we pledge to share datasets and code with the research community.
\end{abstract}

%Experimental results show that \FW attains on the average $15\%$ greater throughput and $3x$ lower latency than the state-of-the-art networking-based approach, ultimately improving by $60\%$ the number of correctly executed tasks. 

\glsresetall

% \begin{IEEEkeywords}
% Mobile Autonomous System, Deep Learning, Edge, Offloading, Experiments, Optimization
% \end{IEEEkeywords}

%The application is clear --- when a rescue team receives the request for intervention, the information available regarding the seriousness of the incident (e.g., the position of the emergency point, its vastness, the number and physical state of people involved) is very limited. Also, it can take a significant time, even hours, before the rescue team can intervene and detect the problems, due to its physical distance from the emergency place. 
%
%
%\Ms (capillary deployed in small towns and cities) 
% %Most importantly, \Ms are also being used to provide first-response aid during emergencies. For example, the Horizon Europe Work Program suggests the use of \Us to support first responders’ operations and to provide the ability to conduct on-scene operations remotely without endangering them \cite{horizon_civil}. 

\section{Introduction}
% ---- Context e motivation  ---- %
\Us, or drones, have obtained significant attention thanks to their potential use in post-disaster scenarios, where human intervention is difficult or inefficient due to the vastness and/or harshness of the area.
%communication infrastructures may be may be precluded, requiring a multi-hop communication. 
% \cite{hayat2016survey}. 
The key advantage of \Us is the combined presence of advanced sensor equipment, %(\textit{e.g.}, cameras, radars, and GPS) 
wireless multi-hop networking and mobility in the same device, thus enabling critical applications such as automatic target (\emph{e.g.,} object, person) detection and tracking. 

To perform their functions, modern \Us necessarily depend on the execution of computation-heavy \gls{dl} tasks to analyze in real time the images of the target area. These tasks  usually rely on very \glspl{dnn} such as ResNet \cite{he2016deep} and DenseNet \cite{huang2017densely}, which are computationally prohibitive for \Us \cite{messous2019game}. To extend \Us battery lifetime and keep task execution time within acceptable levels, offloading the stream of tasks to neighboring edge servers (\textit{e.g.}, the depot) is a feasible option \cite{scherer2020multi, bartolini2021connected, samir2020age,chen2018multiple, wu2018joint, hosseinalipour2020interference}. Unfortunately,  \textit{\Us networks typically experience limited bandwidth and frequent packet loss} \cite{chuprov2022degrading}. Prior work on \U edge task offloading --- discussed in details in Section \ref{related} --- assumes a single-hop communication between the \Us and the edge \cite{yang2020offloading}\cite{callegaro2021seremas}, or focuses only on networking aspects on a multi-hop communication \cite{chen2020vfc}. All fall short in considering the specific task requirements, which ultimately limits the number of correctly executed tasks. Existing work also rarely uses a testbed to measure performance experimentally.

\begin{figure}[h]
\vspace{-0.1in}
\centering
  \begin{minipage}{0.24\linewidth} 
  \begin{overpic}[width=\textwidth]{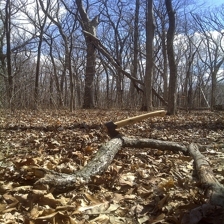} 
   \put (-0.15,0) {{\colorbox{white} A}}
  \end{overpic}
  \end{minipage}  
  \centering
  \begin{minipage}{0.24\linewidth}
    \begin{overpic}[width=\textwidth]{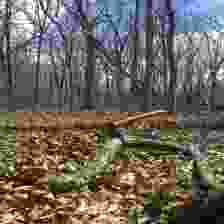} 
   \put (-0.15,0) {{\colorbox{white} B}}
  \end{overpic}
  \end{minipage} 
  \centering
  \vspace{0.04in}
  \begin{minipage}{0.24\linewidth} 
      \begin{overpic}[width=\textwidth]{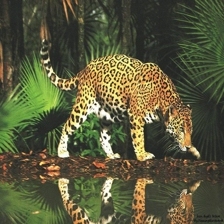} 
   \put (-0.15,0) {{\colorbox{white} C}}
  \end{overpic}
  \end{minipage}  
  \centering
  \begin{minipage}{0.24\linewidth}
        \begin{overpic}[width=\textwidth]{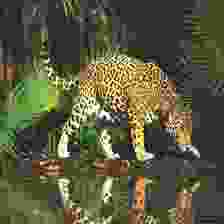}
   \put (-0.15,0) {{\colorbox{white} D}}
  \end{overpic}
  \end{minipage}
  \caption{Examples of application-aware compression. A and C (B and D) represent the images before (after) processing. \vspace{-0.1in}
  }\label{fig:qol_examples}
\end{figure}  

In stark contrast, we propose {\FW}, \emph{an application-aware (A$^2$) optimization framework which jointly optimizes UAV network deployment and task accuracy}. Our key intuition is that a different image compression will result in a different accuracy for the DNN model. Specifically, compressed images have low impact on the network throughput but decrease the DNN accuracy. Conversely, uncompressed images are likely to be classified more correctly but cause higher network overhead. To better highlight this intuition, Figure \ref{fig:qol_examples} shows an example of an original and compressed ($l = 20$) images of \textit{Wildlife} and \textit{Tools} from the ImageNet dataset \cite{deng2009imagenet}. Figure \ref{fig:qol_examples}.B shows how an hatchet at compression level $l =5$ is less recognizable with respect to its original version (Figure \ref{fig:qol_examples}.A), due to the low contrast with the background and the small object dimension. On the other hand, Figure \ref{fig:qol_examples}.D shows how a jaguar at level $l =20$ is easily distinguishable, thanks to its fur pattern and the higher contrast with the background. We consider lossy compression algorithms for images (e.g., JPEG algorithm \cite{bovik2010handbook}), as lossless algorithms do not usually meet the requirements of real-time applications.

\textit{Trading-off network load and latency with task accuracy while aiming at maximizing the target coverage is a daunting challenge.}~\FW considers all these aspects and balances between images compression and network coverage according to current network conditions and application requirements. In Section \ref{sec:qodl} we show that different applications have starkly different compression-accuracy relationships. Specifically, \FW maximizes the number of accomplished tasks at the edge, producing a connected coverage formation for the \Us of the squad and a compression level assignment for the \Us that satisfy the application requirements. The connected coverage formation is designed such that it jointly maximizes the number of covered targets, minimizes network delay due to inefficient routes or channel contention, and minimizes task miss-classification due to high input compression. We show through simulation and \emph{real-world experiments} with a testbed  that \greedyprob attains on the average $38\%$ more accomplished tasks than the state-of-the-art networking-based approach, with a sharp increase ($400\%$ more accomplished tasks) when the number of tasks to offload drastically increases.

\textbf{This paper makes the following novel contributions:}

$\bullet$ We design {\FW} --- a novel application-aware framework that optimizes the number of accomplished tasks at the edge by finding the best network deployment and compression level assignment for the considered application.   
We first design a \textit{Application Aware Task Analyzer} (\texttt{A$^2$-TA}) to learn the requirements of the tasks, and to map the possible data compression of the \Us to the expected task accuracy at the edge. 
Then, we define the \gls{prob} to assign the \Us to the tasks. We prove the NP-hardness of the problem and formulate a polynomial-time \greedyprob algorithm to solve it efficiently; 

$\bullet$ We study the performance of {\FW} with extensive simulations. We analyze six different critical applications for \Us, including Search-and-Rescue, Maritime and Wildlife monitoring. We show how \greedyprob accomplishes around 38$\%$ more tasks with respect to existing network-based approaches, thanks to its application-aware optimization; Whereas, the NP-hard version (\optprob) attains on average 50\% performance increase on restricted problem instances. 
%It shows superior performance thanks to its ability to adapt the collected data to the network condition, optimizing the expected tasks at the edge, rather than only network parameters.  \smallskip

$\bullet$ We implement {\FW} into a testbed and perform \emph{real-field experiments}. We consider four \Us and a Jetson Nano board, mounting a Raspberry PI for computation and communication. We execute an image analysis application in which \Us periodically acquire images from on-board sensors, %at different rate and 
with different delay requirements. We implement four state-of-the-art image classification models (i.e., DenseNet \cite{huang2017densely}, ResNet152, ResNet50 \cite{he2016deep} and MobileNet-V2 \cite{mobilenetv2}), and one object detection model (YoloV4 \cite{Yolov4}) which are executed at the edge server on the Jetson Nano board, and we let the \Us offload tasks through WiFi connection. Experimental results confirm the outstanding performance of \FW. %\vspace{-0.1cm}
% LAVORI MOBICOM - TPC = 
% - Poster: Cell Tower Extension through Drones \cite{dhekne2016poster}
% - Poster: When Autonomous Drones Meet Driverless Cars \cite{wang2018autonomous}
%\\
%Some recent works address these aspects focusing on deployment of the drones over targets or areas of interest \cite{bartolini2020multi, kim2017theoretical, kimura2020distributed}. However, they do not consider possible multi-hop connections among drones to offload collected data, making them unsuitable in absence of long-range radios and with streams of data.  
%
%Similar approaches were proposed for wireless sensor networks, in which sensors aim at monitoring targets while be connected to the base station. A seminal work of Z. Lu et al. in \cite{lu2014maximum} addresses the target coverage and data collection problem. They maximize the network lifetime using repeated sensors, which act both as relay and sensor nodes.  In \cite{liao2014minimizing} Z. Liao et al. address the target coverage and network connectivity in wireless mobile networks, with the goal of minimizing the movement for mobile sensors. They decompose the problem in two sub-problems: the target coverage is solved by a Voronoi-based partition algorithm; while node connectivity is achieved with a constrained edge length Steiner tree.

\section{Related Work}\label{related}
%, he2020beecluster
Only very recently has the literature considered task offloading in the context of edge-assisted \gls{dl}-based applications \cite{callegaro2021seremas, chuprov2022degrading}. Chuprov et al. \cite{chuprov2022degrading} show how the performance of the end-line ML systems is affected by the quality of data and network. They consider packet loss and limited bandwidth in a image classification task, and they recommend to stop the system when packet loss reaches 2-5\%. In Section \ref{sec:exp_section} we show that our system enables the classification task even with $15\%$ of packet loss. Chen et al. \cite{chen2020vfc} consider a hierarchical offloading of computation tasks. Conversely from us, they focus on the communication and routing of tasks toward a more computational powerful device, without focusing on the specific task requirements. Yang  et al. \cite{yang2020offloading} propose a hierarchical \gls{dl} task execution framework, in which only a few lower layers of a \gls{cnn} are on the \Us, while the edge server contains the higher layers of the model, which need more resources. However, a single-hop high-performance 4G network is considered, while we focus on the more challenging scenario of multi-hop connectivity toward the edge. Recently, Callegaro \textit{et al.} \cite{callegaro2021seremas} proposed SeReMAS, a framework  where the application-, network- and telemetry-based features are used to select and assign \Us tasks to the most reliable edge servers. However, a single-hop system is considered, and data compression is not explored. 

As one of the output of {\FW} is a connected coverage formation, we mention also some prior art on \Us deployment optimization \cite{bertizzolo2020swarmcontrol, rashid2020socialdrone, bartolini2020multi, chen2020vfc, wang2019dynamic, kimura2020distributed}, which however does not consider task offloading. Natalizio et al.  \cite{natalizio2019take} have considered the problem of minimizing networking resources while maximizing the user experience (i.e., perceived quality) when filming sport events. Moreover,   \cite{tateo2018multiagent, scherer2020multi, bartolini2021connected} optimize network deployment under continuous or periodic connectivity constraints, but they do not consider critical indicators such as task accuracy with delay constraints, which are critical to the \Us mission. Recently, Nguyen proposed a \gls{stba} for target coverage and network connectivity \cite{nguyen2019new}, where Fermat points and the node-weighted Steiner tree algorithm are used to find a tree such that most of the targets are covered, and the \Us are minimized. In Section \ref{sec:perf_section}, we consider variants of \cite{nguyen2019new} as performance benchmarks for \FW.

\section{The {\FW} Framework}\label{sec:framework}
In this section, we give an overview of \FW (Section \ref{sec:taskmas}) and describe the two key components of \FW:  \texttt{A$^2$-TA} (Section \ref{sec:qodl}) and \gls{prob} (Section \ref{sec:prob_stpp}).

\begin{figure}[h]
\centering
\vspace{-0.1cm}
\includegraphics[width=\linewidth]{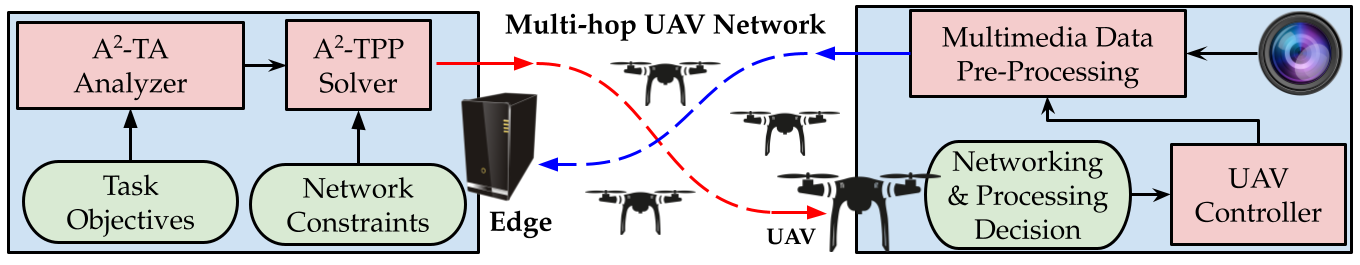}
% https://docs.google.com/drawings/d/1WODwzicMKWMYEBY_J-mBBqNqI0zZw5jw5gcIz0Dfn5k/edit?usp=sharing
\caption{High-level overview of {\FW}.%\vspace{-0.2cm}
}
\label{fig:framework}
\end{figure}

Figure \ref{fig:framework} shows a high-level overview of \FW. The Application-Aware Task Analyzer (\texttt{A$^2$-TA}) at the  \textit{edge server} learns the relationship between the image compression and the accuracy on a set of classes of interest, and passes its output to the Application-Aware Task Planning Problem (\gls{prob}) solver, which jointly optimizes \Us positions, routing policy, and data compression strategy to maximize the number of correctly executed tasks per unit of time. The optimal network deployment and image compression levels are sent to the \textit{\Us network}, which moves to the targets, monitors them and streams tasks to the edge through the multi-hop connection. \vspace{-0.1cm}

% TARGETS --> Event of interest to monitor
% TASKS -> deep learning computation generated by a target to be offloaded.

\subsection{System Model and Assumptions}\label{sec:problem}
We assume one or more \Us are deployed over an \gls{aoi}, which contains several \textit{targets}, e.g.,  the location of a vehicle, person, or any entity of interest.  Each \Us is equipped at least with (i) multimedia sensors (e.g., camera and microphone); (ii) a single radio for communication; and (iii) a computational unit. The edge server is equipped with low-latency hardware for \gls{dl} computation. We do not rely on any communication infrastructures (e.g., 5G) and assume edge offloading is realized through multi-hop communication. A \U monitors a target by sampling data through its sensors and generates a \textit{task} to be executed at the edge. A task could be ``car, bicycle, or bus detection on a video camera frame every 10 frames''. The task is then sent to the edge server through a multi-hop connection, where a state-of-the-art \gls{dl} model is run to perform the task. We assume each task has mission-driven constraints in terms of (i) minimum classification accuracy given a specified \gls{dl} model; (ii)  maximum latency, defined as the time between the task generation and its successful execution. Thus, a task is successfully executed if (i) promptly offloaded to the edge;  and (ii) correctly analyzed by the model within a deadline.
\begin{table}[h]
\centering
\footnotesize
\begin{tabular}{|p{0.05\textwidth}|p{0.38\textwidth}|}
\hline
\textbf{Notation} & \textbf{Description} \\ \hline
   $\mathcal{U}$  & A set of available \Us of the fleet\\ \hline
   $\mathcal{T}$  &  A  set of targets to cover \\ \hline
   $\sigma$       &  The edge server      \\ \hline
   $r_{\text{com}}$          &    Drone's communication radius \\ \hline
   $r_{\text{sens}}$         &  Drone's sensing radius \\ \hline
   $l_u$ &    Distance traveled by a drone \\ \hline
   $\delta_u$ &  Drone's overall energy consumption \\ \hline
    $\hat{e}^s_{ij}$ & Amount of data transmitted through the link between \U $i$ and \U $j$\\ \hline 
    $\hat{e}^a_{ij}$ & Expected task accuracy at the edge, for each task \\ \hline
   $\rho_{i,j}$ & Estimated channel data rate \\ \hline
   $\hat{p}_u$ & Position vector assigned by the solver to drone\\ \hline
   $\beta_u$ & Energy spent for each distance unit  traveled at constant speed\\ \hline 
   $\alpha_u$ & Energy spent in a steady position\\ \hline
   $\hat{\omega}^u_{ij}$ & Drone $u$ monitors the target $i$ with a compression $j$ or not\\ \hline
   $\phi_u$ & \U initial energy\\ \hline
   $Q(s, l)$ & A tuple with expected accuracy and data size of the frame of the application scenario $s$, with compression level $l$  \\ \hline 
   $\Psi$ &    The final connected coverage formation returned by the greedy algorithm \\ \hline
   $\tau_\text{best}$ &    Best coverage found during an iteration \\ \hline
   %$\tau_\text{temp}$ &    Temporary coverage generated to find the best target to add \\ \hline
   $\tau_\text{par}$ &    Partial coverage to enhance or join with $\Psi$ \\ \hline
\end{tabular}
\caption{Table of Symbols.\vspace{-0.5cm}}
\end{table}

\subsection{Overview of \FW}\label{sec:taskmas}

The ultimate goal of  \FW is to maximize the number of correctly executed tasks. To approach this challenging issue, and conversely from existing work, \textit{\FW takes into account how the task success is affected by the image compression}. To this end, \FW jointly optimizes the deployment of \Us and the task offloading  to maximize the number of executed tasks. Each \U is made up of two key modules. First, the \textit{\U Controller} implements networking and data processing decisions (next position, targets to cover, sampling process, and offloading routes) received from the \gls{prob}. Second, the \textit{Multimedia Data Pre-Processing} module samples data and creates tasks by pre-processing collected multimedia data according to the \gls{prob} solution. %This, in turn, helps achieve the right trade-off between network load and expected accuracy at the edge. The goal is to compress tasks to reduce network usage, while still achieving good accuracy within delay constraints. 

\subsection{Application-Aware Task Analyzer (\texttt{A$^2$-TA})}\label{sec:qodl}

The \texttt{A$^2$-TA} module determines the relationship between different image compression levels and task accuracy %to reduce communication overhead while keeping satisfying task accuracy. 
so that \Us can reduce the amount of transmitted data and avoid network congestion, while satisfying application requirements. 
%Existing approaches neglect the task accuracy and optimize only networking-based metrics such as delay and throughput, which causes over-provisioning as shown in Section \ref{sec:perf_section}. {\color{red} MP: explain why or remove the sentence.}

\begin{figure}[h]
\centering
%  trim={<left> <lower> <right> <upper>}
\includegraphics[width=0.9\linewidth]{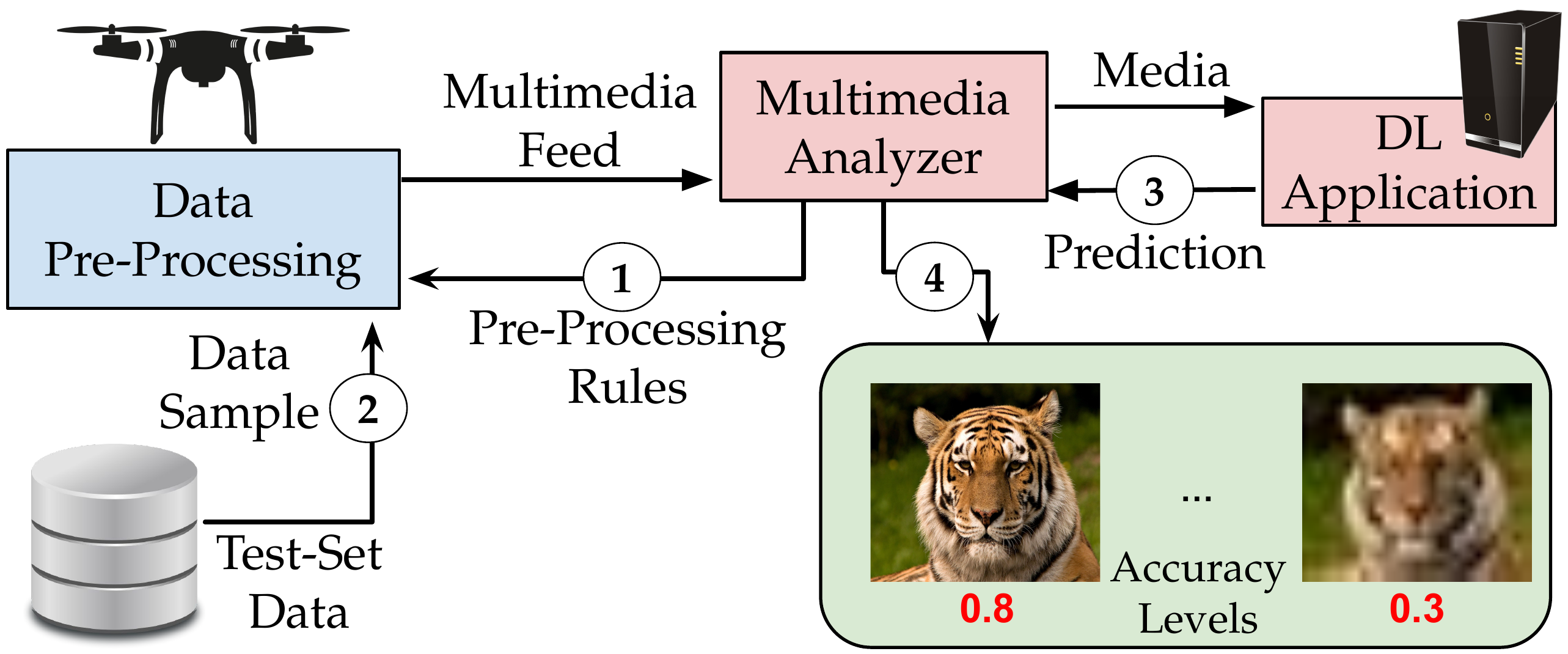}
% https://docs.google.com/drawings/d/18VjB684eYMcfG3GqPKzg9lB0mH3huR8RxhVvk_1eVNE/edit?usp=sharing
\caption{Main Operations of the \texttt{A$^2$-TA} module. \vspace{-0.1cm}}
\label{fig:qol_analyzer}
\end{figure}

Figure \ref{fig:qol_analyzer} shows the workflow of the \texttt{A$^2$-TA}, which is executed before the network deployment, by using the datasets of the specific \gls{dl} application. The \texttt{A$^2$-TA} iterates over the sampled data, changing compression according to the \U sensors capabilities (\textbf{step 1} and \textbf{step 2}). Each revised data sample is fed to the \gls{dl} application, which outputs a prediction (\textbf{step 3}). Finally, predictions are compared with the ground truth, to infer and store the model performance according such data compression (\textbf{step 4}).

More formally, the function: $ Q : \mathcal{S} \times \mathcal{L} \rightarrow \mathbb{R}^2$  maps an application scenario in the set $\mathcal{S}$ and a compression level in the discrete set $\mathcal{L} = \{1, .., 100\}$ a tuple in $\mathbb{R}^2$ representing the average \textit{accuracy} and \textit{data size}. The function is determined by iterating over the application scenarios, possible compression levels and relative dataset entries. Each sampled image is compressed according the compression level $l$ by the JPEG compression algorithm \cite{bovik2010handbook}, and it is fed into the \gls{dl} model to determine accuracy and data size.
% Algorithm \ref{alg:compressioalgo}
%
%Notice that, the \textit{accuracy} defines the expected performance of \gls{dl} model with $l$ resolution, while the \textit{data size} accounts for the expected amount of bytes transmitted in the network.

\begin{figure}[h]
\centering
  \begin{minipage}{0.48\linewidth} 
  \begin{overpic}[width=\textwidth]{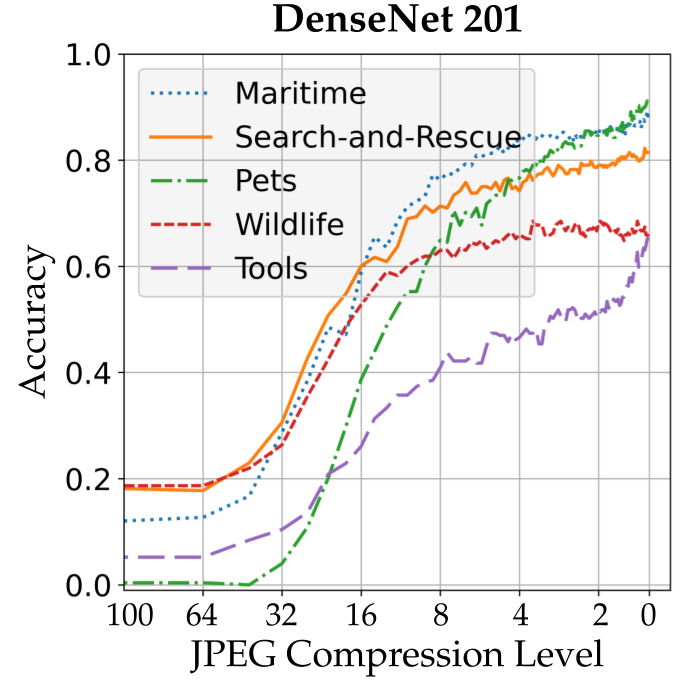} 
  \put (80,20) {{(a)}}
  \end{overpic}
  \end{minipage}  
  %\centering
 % \begin{minipage}{0.49\linewidth}
 % \begin{overpic}[width=\textwidth]{img/qol/17_01_22/final_accuracyDenseNet201.png}
 % \put (80,20) {{(b)}}
 % \end{overpic}
 % \end{minipage} 
 % \centering
  \begin{minipage}{0.48\linewidth} 
  \begin{overpic}[width=\textwidth]{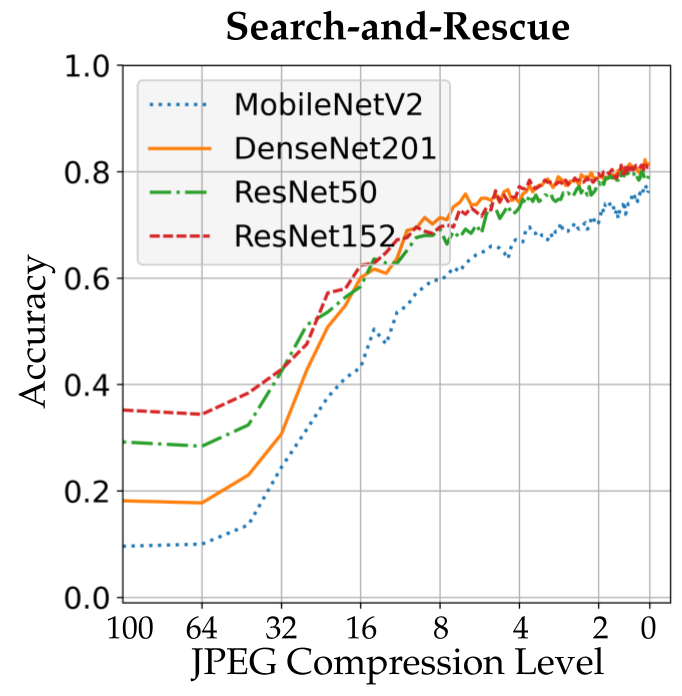} 
  \put (80,20) {{(b)}}
  \end{overpic}
  \end{minipage}  
  %\centering
  %\begin{minipage}{0.32\linewidth}
  %\begin{overpic}[width=\textwidth]{img/qol/17_01_22/final_size.png}
  %\put (80,20) {{(c)}}
  %\end{overpic}
  %\end{minipage}
  %
  \caption{ Accuracy as a function of JPEG compression level.}\label{fig:qol_levels}
\end{figure}

To give an example, we consider 5 scenarios: \textbf{Maritime} (Fireboat, Wreck, Lifeboat, Ocean liner, Speedboat),  \textbf{Search-and-Rescue (SaR)} (Fire truck, Ambulance, Police van, German shepherd, Pickup truck) \textbf{Wildlife} (Kit fox, Polecat, Red wolf, Zebra, Jaguar), \textbf{Tools} (Screwdriver, Power drill, Hatchet, Hammer, Chainsaw),  \textbf{Pets} (Golden retriever, Pomeranian, Guinea pig, Persian cat, Hamster). Figure \ref{fig:qol_levels}.a shows the accuracy for the different scenarios as a function of the compression,  while Figure \ref{fig:qol_levels}.b shows the accuracy for the same scenario when different DL models are used. The figures  highlight the need of the \texttt{A$^2$-TA}. For example, images of \textit{Tools} have low accuracy, constraining the compression at level $l=8$ to achieve at least $40\%$ of accuracy, while \textit{Wildlife} can achieve the same accuracy with higher compression $l=25$.

\subsection{\gls{prob} MILP Formulation}\label{sec:prob_stpp}

%{\color{red} MP: Mettere da qualche parte:}
%In this paper we address the problem of designing connected coverage formations for fleet of \Us, which are employed in 
%
We  formalize \gls{prob} as a Mixed Integer Linear Programming (MILP). Hereafter we denote the edge server as $\sigma$, the set of targets to monitor as $\mathcal{T}$, the set of available \Us of the fleet as $\mathcal{U}$. The \gls{prob} solver outputs: (i) a connected coverage formation of \Us, 
%(ii) multi-hop routing paths for the \Us monitoring targets to $\sigma$ 
and (ii) the compression level each drone must adopt to capture images when inspecting a target. 

%The objective of \gls{prob} is the accomplished tasks maximization. Here we formally define the concept of connected coverage formation and accomplished task. 

\begin{definition}
\label{def:DeploymentTree}
A subset of \Us $F \subseteq \mathcal{U}$ is deployed according to a \textbf{connected coverage formation}, when some of the \Us in $F$ are employed to inspect a subset of targets $M \subseteq \mathcal{T}$, while being connected to the base station $\sigma$, either directly or through a \mbox{multi-hop sequence of the other \Us in $F$}.
\end{definition}

\begin{definition}
\label{def:task}
A \textbf{task} for a drone $u \in F$ covering a target $t \in M$, consists in delivering an image captured  from the drone's on-board cameras to the base station $\sigma$. It is said to be an \textbf{accomplished task} if two conditions hold: (1) when received at $\sigma$, the time since the task was created is not superior to a threshold $\Delta$; (2) when the captured image reaches $\sigma$, the DL model outputs a correct prediction for it.
\end{definition}

We define the UAV sensing range and communication range as $r_\text{sens}$ and $r_\text{com}$, respectively. 
%We let $\mathcal{T'} = \mathcal{T} \times \mathcal{L}$ be the set of variables $\hat{t}^i_j$ indicating that the target $i \in \mathcal{T}$ was monitored with resolution $j \in \mathcal{L}$.  
We denote with $p$ the position vector of the entities involved, and to $(p_u^x, p_u^y)$ for the $x$ and $y$ coordinates respectively.  In particular $p_u$ is the position of UAV $u, \; \forall u \in \mathcal{U}$ at the beginning of the mission; $p_i$ to the position of the target $i, \; \forall i \in \mathcal{T}$; $p_\sigma$ the position of the edge server. We denote with $\hat{p}_u$ the position vector assigned by the solver to drone $u, \; \forall u \in \mathcal{U}$. We define the distance traveled for each drone as $ l_u = |p_u - \hat{p}_u|, \ \forall u \in \mathcal{U}$.

%\textbf{Problem Constraints.}~
To estimate energy consumption, we define $\beta_u$ as the energy spent for each distance unit  traveled at constant speed, and $\alpha_u$ as the energy spent in a steady position, for a given unit of time. The values of $\beta_u$ and $\alpha_u$ are estimated through on-field experiments or from technical specifications. The overall energy consumption for UAV $u$ is defined as \mbox{$ \delta_u = l_u \cdot \beta_u + \alpha_u \cdot \Lambda, \ \forall u \in \mathcal{U}$}, where $\Lambda$ is an upper-bound of the time required, once reached the targets, to monitor the targets and complete the mission. We constraint the reachable points according to the \Us initial energy $\phi_u$, as follows:
\begin{equation}
    \delta_u +|p_\sigma - \hat{p}_u| \cdot \beta_u \leq \phi_u, \ \forall u \in \mathcal{U}
\end{equation}
This constraint defines the positions that are reachable as they let the \Us with enough energy to come back to the edge-server for recharging operations. We use the binary variables $\hat{\omega}^u_{ij} \in \{0, 1\}$, which define if the drone $u$ monitors the target $i$ with compression level $j$ or not. A target $i$ is monitored if an only if the \U $u$ is close enough to the target position $p_i$. Formally we want to constraint $\hat{\omega}^u_{i,j} = 1 \iff |\hat{p}_u - p_i| \leq r_\text{sens}$ which becomes $\forall u \in \mathcal{U}, \; \forall i \in \mathcal{T}, \; \forall j \in \mathcal{L}$: : 
\begin{equation}
\begin{split}
r_\text{sens} & \geq |\hat{p}_u - p_i| - M_{\text{cost}} \cdot (1 - \hat{\omega}^u_{i,j}) \\ 
r_\text{sens} & \leq |\hat{p}_u - p_i| + M_{\text{cost}} \cdot \hat{\omega}^u_{i,j} \\ 
\end{split}
\end{equation}
where $M_{\text{cost}}$ is a big constant number. Next, we enforce that a target is covered by at most one UAV and that each UAV can cover only one target:
\begin{eqnarray}
    \sum\limits_{u \in \mathcal{U}} 
    \sum\limits_{j \in \mathcal{L}} 
    \hat{\omega}^u_{ij} \leq 1, \forall i \in \mathcal{T}, \ \ 
    \sum\limits_{i \in \mathcal{T}} 
    \sum\limits_{j \in \mathcal{L}}
    \hat{\omega}^u_{ij} \leq 1, \forall u \in \mathcal{U}
\end{eqnarray}

\noindent We introduce an extended node set $\mathcal{V} = \mathcal{U} \cup \{\sigma\}$ and we consider all the possibles communication paths among nodes, i.e., all the edges $(i, j) \; \forall i, j \in \mathcal{V}$. A binary variable $\hat{\gamma}_{i,j} \in \{0, 1\}$, indicates if the nodes $i$ and $j$ are too distant to communicate. The relation $|\hat{p}_i - \hat{p}_j| \geq r_\text{com}  \Rightarrow \hat{\gamma}_{ij} = 0 $ is enforced as follows: 
\begin{equation}
\begin{split}
|\hat{p}_i - \hat{p}_j| \leq r_\text{com} + M_{cost} \cdot (1 - \hat{\gamma}_{ij}), \ \forall i, j \in \mathcal{V}
\end{split}
\end{equation}

%\textbf{Problem Formulation.}~

\noindent We define the data frame offloading as a network flow formulation. We introduce a set of variables $\hat{e}^s_{ij}$ defining the amount of data transmitted through the link between \U $i$ and \U $j$. We define $\hat{e}^a_{ij}$ to account for the expected task accuracy at the edge, for each task. We impose that the edge does not generate any outgoing flow, for both data and accuracy flows:
\begin{equation}
    \sum\limits_{j \in \mathcal{V}} \hat{e}^s_{\sigma j} \leq 0, \ \ \  \sum\limits_{j \in \mathcal{V}} \hat{e}^a_{\sigma j} \leq 0 
\end{equation}

\noindent We allow a flow only for between neighboring nodes:
\begin{equation}
    \hat{e}^s_{ij} + \hat{e}^a_{ij} \leq \hat{\gamma}_{ij} \cdot M, \ \forall i, j \in \mathcal{U}
\end{equation}

\noindent The maximum bandwidth allowed between two \Us is constrained to respect the estimated channel data rate $\rho_{i,j}$:
\begin{equation}
    \sum\limits_{j \in \mathcal{V}} \hat{e}^s_{ij} \leq \rho_{i,j},  \ \forall i \in \mathcal{U}
\end{equation}
We specify that a \U can transmit only towards another \U, resulting into a tree rooted at the edge:  
\begin{equation}
    \sum_{j \in \mathcal{V}} \hat{\gamma}_{ij} \leq 1, \ \forall i \in \mathcal{U}
\end{equation}
We impose flow conservation as follows:
\begin{equation}
    \sum\limits_{k \in \mathcal{V}} \hat{e}^s_{uk} - \sum\limits_{k \in \mathcal{V}} \hat{e}^s_{ku} = 
    \sum_{i \in \mathcal{T}} \sum_{j \in \mathcal{L}} b_{i,j} \cdot \hat{\omega}^u_{ij}, \ \forall u \in \mathcal{U}
\end{equation}
which imposes that, for each outgoing edge from $u$, the flow is increased by expected data size of the target covered by the \U $u$ . We also impose that the edge receives all the data produced by the covered targets:
\begin{equation}
    \sum\limits_{k \in \mathcal{V}} \hat{e}^s_{k \sigma} = \sum_{i \in \mathcal{T}} \sum_{j \in \mathcal{L}} b_{i,j} \cdot \hat{\omega}^k_{ij}
\end{equation}

\noindent To conclude, we constraint the accuracy of the targets at the edge-server, as follows:
\begin{equation}
    \sum\limits_{k \in \mathcal{V}} \hat{e}^a_{uk} - \sum\limits_{k \in \mathcal{V}} \hat{e}^a_{ku} = 
    \sum_{t^i_j \in \mathcal{T'}} a_{i,j} \cdot \hat{\omega}^u_{ij}, \ \forall u \in \mathcal{U}
\end{equation}

\noindent \textbf{Objective Function:} maximize covered targets, \gls{dl} tasks accuracy, and energy spent by the \Us: 
\begin{equation}\label{eq:lp_function}
    \max \ \ \alpha \cdot \sum\limits_{j \in \mathcal{V}} \hat{e}^a_{\sigma j}  +
    \beta \cdot \sum\limits_{i \in \mathcal{T}, j \in \mathcal{L}, u \in \mathcal{U}} \hat{\omega}_{ij}^u -
    \eta \cdot \sum\limits_{u\in \mathcal{U}} l_u 
\end{equation}
The term $\alpha$ prioritizes the maximization of the accuracy, while $\beta$ weights the importance of covering the targets and $\eta$ minimized the distance traveled by the \Us. 

\begin{theorem}\label{th:np_hardness}
The \gls{prob} problem is NP-Hard.
\end{theorem}

\begin{proof}
We show that \gls{prob} generalizes the \textit{Steiner tree problem with minimum number of Steiner points and bounded edge-length} STPMSPBEL, a known NP-hard problem \cite{lin1999steiner}. Given a set $P$ of $n$ terminal points in a 2-dimensional plane, a positive constant $R$, and a non-negative integer $B$, STPMSPBEL asks whether it exists a tree spanning a set of points $P \subseteq Q$ s.t. each edge has a length less than $R$ and the number of Steiner points (i.e., $Q \setminus P$) is less than or equal to $B$. Any instance of STPMSPBEL can be reduced to an instance of our problem in polynomial time. The set of points $P$ represents our target set $\mathcal{T} \cup \{ \sigma \}$, and $B$ defines the number of available \Us, with communication range equal to $R$. We consider \Us with unlimited batteries  and one compression level (i.e, $|\mathcal{L}| = 1$). This problem instance finds a solution that maximizes the number of connected targets with the edge server, moving the minimum number of \Us. If such a solution exists, and covers all the points in $P$, then it also exists a tree spanning a set of points $P \subseteq Q$, where each edge has length less than $R$ and the number of Steiner points is less then or equal to $B$. The complexity of the above reduction is polynomial, thus we derive that \gls{prob} problem is at least as hard as the STPMSPBEL problem \cite{lin1999steiner}.
\end{proof}

%Our formulation has the benefit of using an Euclidean-distance approximation function, which keeps the optimization model linear. Indeed, in the performance evaluation we never experienced computational time longer than 15 minutes. Based on our experience this time is less than the typical time required to set up a network of \Ms in a real testbed. We executed the model over a laptop with an Intel i7-8665U CPU and 32GB of RAM, using Gurobi solver \cite{gurobi}. 

\section{A Polynomial Time Heuristic for \gls{prob}} \label{sec:poly} 

We propose a greedy heuristic to solve \gls{prob} in polynomial time. We first introduce the algorithm, and then prove its polynomial time complexity.

\subsection{Algorithm Overview}

\noindent \textsc{\greedyprob} outputs a \textit{connected coverage formation} --  also referred to as \textit{coverage} for brevity -- for the \Us, and a \textit{compression level assignment} for each covered target. 
Both coverage and compression need to meet the criteria expressed in Equation \ref{eq:lp_function}, that is, optimizing the \textit{number of accomplished tasks}. Our approach is to maximize the number of inspected targets, producing a coverage of minimum congestion, and minimizing task misclassification due to low frame resolution. 
%In particular, a series of local optimizations is done so to contribute to that main objective. 
%According to Definition \ref{def:task}, it is paramount to produce a connected coverage formation that: maximizes the number of inspected targets, minimizes network delays due to inefficient routes or channel contention, and minimizes task miss-classification due to low input resolution. 
%The algorithm relies on the Triangular Stainer Tree (\texttt{TST}) problem \cite{senel2011relay}.
%a variant of the well-known Steiner Tree problem \cite{hwang1992steiner}. 
%Given a set of vertices, the \texttt{TST} finds a tree connecting them through additional points, s.t. the sum of the edges length is minimized, while being subjected to a bound on the edge length.
%In our case terminals are the targets and Steiner points are relay nodes or targets nodes, whereas the bound on the edges is given by the communication range o the \Us $r_\text{com}$.
%Thus, in our case, a connected coverage formation is a Triangular Stainer Tree, and the bound on the edges length is given by the communication range $r_\text{com}$ of the \Us . 
Specifically, a coverage $\tau = (V, E, W)$ is a Triangular Steiner Tree \cite{senel2011relay} in which the set of nodes $V$ represents the positions \Us must reach to cover the target nodes in $M$, while staying connected with the base station $\sigma$ in a multi-hop manner. The set $E$ represents the link between \Us, thus the routes data streams must follow through the network. The function $W$ maps each $e \in E$ to a weight that represents link's bandwidth. In our implementation we estimate this value empirically.
% %
It is assumed that at the base station, communication happens through dedicated transceivers and does not require actual coverage with a drone. Thus, at any time it holds $|V| \leq |\mathcal{U}|+1$. 
\subsection{\greedyprob}

\setlength{\textfloatsep}{4pt}% Remove \textfloatsep
\begin{algorithm}[t]
{

\footnotesize

\SetAlgoCaptionSeparator{:}
\KwIn{$\mathcal{U}$: set of \Us, $\mathcal{T} $: set of targets} 
\KwOut{$\Psi$ a connected coverage formation}

\vspace{.2cm}

$\widehat{\mathcal{T}}, \Psi, \tau_\text{par}, c_\text{par} \leftarrow \{\sigma\}, \{\sigma\}, \{\sigma\}, 0$ \\
\While{\normalfont $ \mathcal{T} - \widehat{\mathcal{T}} \neq \emptyset$ \textbf{or}  $ |V_{\Psi} \cup V_\text{par}| < |\mathcal{U}| $} { 
$t_\text{best}, \tau_\text{best}, c_\text{best} \leftarrow \emptyset, \emptyset, \infty$ \\
\For{$t \in \mathcal{T} - \widehat{\mathcal{T}}$} {
    
    $\tau_\text{temp} \leftarrow  \text{\textsc{TST}}(\{ t \} \cup \mathcal{C}(\tau_\text{par}), r_\text{com})$ \\
    
    $c_\text{temp} \leftarrow \textsc{Cost}_\alpha(\tau_\text{temp}, \tau_\text{par}, \Psi,  \textsc{Compression}(\tau_\text{temp}))$ \\
    
    \If {\normalfont $c_\text{temp} < c_\text{best}$ \textbf{and} $|V_\text{temp}|-1 \leq |\mathcal{U}| - |V_{\Psi} \cup V_\text{par}|$} {
        $t_\text{best}, \tau_\text{best}, c_\text{best} \leftarrow t, \tau_\text{temp}, c_\text{temp} $ \\
    }
}
\If {\normalfont $t_\text{best} = \emptyset$} { 
    $\Psi \leftarrow \Psi \cup \tau_\text{par}  $ \\
    \textbf{break}
}

$\tau_\text{los} \leftarrow \text{\textsc{TST}}(\{\sigma, t_\text{best}\}, r_\text{com})$ \\

$c_\text{los} \leftarrow$  $\textsc{Cost}_\alpha( \tau_\text{los}, \tau_\text{par}, \Psi, \textsc{Compression}(\tau_\text{los}))$\\

\If {\normalfont $c_\text{los} < c_\text{best} - c_\text{par}$} {
   $ \Psi \leftarrow \Psi \cup \tau_\text{par}$ \\
   $ \tau_\text{par}, c_\text{par} \leftarrow \tau_\text{los}, c_\text{los} $ \\
}
\Else {
    $\tau_\text{par}, c_\text{par} \leftarrow \tau_\text{best}, c_\text{best} $ \\
}
$\widehat{\mathcal{T}} \leftarrow \widehat{\mathcal{T}} \cup \{ t_{\text{best}} \}$ \\
}
$R, \cdot \leftarrow \textsc{Compression}(\Psi)$ \\
\textbf{return} $\Psi$, R
}
\caption{\textsc{Greedy \gls{prob}}}
\label{algo:main}
\end{algorithm}  

Algorithm \ref{algo:main} returns a coverage formation $\Psi$, merging partial coverage formations $\tau_\text{par}$ generated to cover targets using the minimum deployment cost at each iteration. In the initialization phase, we let: $\mathcal{\widehat{T}}$ be the set of covered targets, initially containing only the base station;
%as it is the only target not requiring coverage by a drone; 
$\Psi$ be the coverage archived so far;
%, to which branches rooted at the base station are merged;
$\tau_\text{par}$ be the partial coverage iteratively grown that is added to $\Psi$ when it cannot be further expanded; $c_\text{par}$ be the cost of the partial coverage generated so far (\textbf{line 1}). 
The while loop iterates until either all the targets are covered $ \mathcal{T} - \widehat{\mathcal{T}} \neq \emptyset$ or the number of \Us used does not exceed the fleet size (\textbf{line 2}).
%in the archived connected coverage formation $V_\text{tot}$ and the tree that is currently growing $\tau_\text{loc}$ do not exceed in number the fleet size (\textbf{line 2}). 
%In the first part of the loop, the best candidate target to cover $t_\text{can}$ is chosen. 
The variables $t_\text{best}, \tau_\text{best}, c_\text{best}$ contain respectively the best target found at each iteration, the coverage including that target and its cost. A for loop over the uncovered targets  $\mathcal{T} - \widehat{\mathcal{T}}$ allows to find the best target to add, building new temporary coverage formations $\tau_\text{temp}$ using the targets already covered by $\tau_\text{par}$ (namely the set $\mathcal{C}(\tau_\text{par})$) and adding to them the candidate target $t$. Then we evaluate the cost of $\tau_\text{temp}$ (\textbf{lines 3-6}). This cost combines the number of drones needed for the coverage, and the loss in accuracy due to the channel contention. We will talk in more detail about how this cost is computed when describing Algorithm \ref{alg:phase2}. Then we check if: (i) $\tau_\text{temp}$ has a lower cost than $\tau_\text{best}$ (ii) and if $\tau_\text{temp}$ can be covered with the remaining \Us.
%\Us $|V_\text{temp}|-1$ inferior to the currently available \Us (i.e., $|\mathcal{U}| - |V_\text{tot} \cup V_\text{loc}|$), 
If both checks go through, then $t$ becomes the best candidate $t_\text{best}$  and the associated candidate coverage $\tau_\text{temp}$ with its cost $c_\text{temp}$ are stored into $\tau_\text{best}$ and $c_\text{best}$ respectively (\textbf{lines 7-8}). In case no target was set as a best candidate, (i.e., $t_\text{best} = \emptyset$) the while loop breaks. This happens only when the second condition at line 7 is not met for any target, that is no coverage formations can stick to the remaining fleet size constraint. 
%In fact, condition one is always met because, as previously discussed, the cost of a connected coverage formation cannot be infinite. 
%Thus when no drone is available the algorithm jumps to line 20.
Then, the cost paid to cover only $t_\text{best}$ that is $c_\text{best} - c_\text{par}$ is compared to the cost $c_\text{los}$ of a new line-of-sight (\textit{los}) branch $\tau_\text{los}$ grown using only $t_\text{best}$ as target. If $\tau_\text{los}$ costs less than the partial grown tree so far $\tau_\text{par}$, then $\tau_\text{par}$ is merged with the final tree $\Psi$. Then $\tau_\text{los}$ becomes the new partial connected coverage to grow. Otherwise growing $\tau_\text{par}$ is still convenient, so $\tau_\text{best}$ becomes the new partial deployment including the new target $t_\text{best}$ and $\tau_\text{los}$ is discarded (\textbf{lines 12-19}). 
When the algorithm terminates (\textbf{line 20}) the final coverage $\Psi$ along with all the compression levels assigned to each target are returned.

\setlength{\textfloatsep}{4pt}% Remove \textfloatsep
\begin{algorithm}[t]
{
\footnotesize
\KwIn{ a coverage formation $\tau_i$}
\KwOut{$R$ vector with compression levels for all targets in $\tau_i$, $L$ vector with loss in accuracy due to all targets in $\tau_i$}

\vspace{.2cm}

{
sort $t$ by $Q(\mathcal{S}(t),*).b \; \forall t \in \mathcal{C}(\tau_i)$ in ascending order \\

$\widehat{\mathcal{C}}, L, R \leftarrow \mathcal{C}(\tau_i), \langle \rangle, \langle \rangle $ \\

\vspace{.15cm}
\For{\normalfont $ t \in \widehat{\mathcal{C}}$} {
\setstretch{1.2}
    %$S \leftarrow \texttt{DFS}(\tau_i, \widehat{\mathcal{C}})$ \\
    $P \leftarrow \textsc{Shortest-Path} (\tau_i, \sigma, t)$ \\
    $B \leftarrow \textsc{Bottleneck} (\tau_i, P) $ \\
    %$N \leftarrow \min \{ B; \; Q(t, *).b \} $ \\
    $R(t) \leftarrow \arg \max_{l \in \mathcal{L}} Q(\mathcal{S}(t), l).b \; \leq  \min \{ B; \; Q(\mathcal{S}(t), *).b \} $ \\
    $L(t) \leftarrow  Q(\mathcal{S}(t), *).a - Q(\mathcal{S}(t), R[t]).a$ \\
    %$$R \leftarrow \arg \max_l \frac{B}{Q(t, l).b}\cdot Q(t, l).a$$ \\
    $W_{i}(e) \leftarrow W_{i}(e) - Q(\mathcal{S}(t), R(t)).b \quad \forall e \in P$ \\
    $\widehat{\mathcal{C}} \leftarrow \widehat{\mathcal{C}} - \{t\}$ \\
}
}
\vspace{.15cm}

\textbf{return} $R, L$
}

\label{algo_line:end2}
\caption{ \textsc{Compression}}
\label{alg:phase2}
\end{algorithm}

% _______________________NEW VERSION
\begin{comment}

\setlength{\textfloatsep}{4pt}% Remove \textfloatsep
\begin{algorithm}[t]
{
\footnotesize
\KwIn{A coverage formation $\tau_i$, A scenario typology $s\in \mathcal{S}$}
\KwOut{The cost $c_i$ of $\tau_i$}
{
$ accuracy_{loss} \leftarrow 0$\\
$uav_{\text{cost}} \leftarrow \frac{|V_i - V_{i-1}|-1}{|\mathcal{U}| - |V_{\Psi} \cup V_{i-1}|} $\\
$accuracy_{\text{max}}, size_{\text{max}} \leftarrow Q(s, *) $\\
\For{\normalfont $ t \in {\mathcal{C}(\tau_i)}$} {
\setstretch{1.2}
    $path_{\sigma,t} \leftarrow \texttt{stp} (\tau_i, \sigma, t)$\\
    $e_{\text{bn}}, w_{\text{bn}}, nodes_{\text{bn}} \leftarrow \texttt{bottleneck}(path_{\sigma, t}, \tau_i)$ \\
    %rapporto tra w e nodes
    $size_{\text{bn}} \leftarrow min(\frac{w_{\text{bn}}}{nodes_{\text{bn}}}, size_{\text{max}})$ \\
    $accuracy_{t}, l_t \leftarrow Q^{-1}(s, size_{\text{bn}})$ \\
    %$$R \leftarrow \arg \max_l \frac{B}{Q(t, l).b}\cdot Q(t, l).a$$ \\
    $accuracy_{\text{loss}} \leftarrow accuracy_{\text{max}} - accuracy_{\text{t}}$ \\
    $\mathcal{R}_{\tau_i}(t) \leftarrow l_t$ \\
    $\mathcal{W}_{\tau_i}(e_{\text{bn}}) \leftarrow w_{\text{bn}} - size_{\text{bn}}$ \\

}
}

\textbf{return} $\alpha * (\frac{accuracy_{\text{loss}}}{|\mathcal{C}(\tau_i)|}) + 1-\alpha * (uav_{\text{cost}})$
}

\caption{ \texttt{cost$_\alpha$} }
\label{algo:phase2}
\end{algorithm} 
\end{comment}

\subsection{Assignment of Compression Levels}

%{\color{red} MP: mail Francesco: dire che algoritmo 1 serve a valutare una connected coverage formation prodotta da algoritmo 2.} 
Algorithm 2 determines the compression levels for each UAV in $F$ inspecting targets in $M$. The rationale is to increase the compression of data flowing from a target, based on the bandwidth assigned to it, leaving more bandwidth to targets having more to send. The algorithm iterates over the targets $t$ covered by the candidate input tree, sorted in ascending order based $Q(\mathcal{S}(t),*).b$, that is the load produced by the target $t$ according to the task analyzer \texttt{A$^2$-TA}, belonging to the application scenario $\mathcal{S}(t)$ and at the minimum compression level (denoted by $*$)  (\textbf{lines 1-3}). 
At each iteration a bottleneck bandwidth $B$ for the target is computed. This quantity is the bottleneck capacity on the path from the source of flow $t$, to the destination $\sigma$. This value is influenced by the number of targets $t$ shares this path with. The bandwidth allocation function can be thought as slight modification of the Depth First Search (\textsc{DFS}) (\textbf{line 5}). %, when all the neighbors of a node are visited, the percentage bandwidth split of the edges is updated according to the targets that share the edge. 
%
%To measure 
%To estimate the quantity of flow available to the target $t$, the minimum flow on the shortest path (\texttt{STP}) from $t$ to $\sigma$ in tree $\tau_i$
% simple path from $t$ to $\sigma$ 
%is measured (\textbf{lines 5-6}). 
To derive the maximum quantity of load that can be transferred from the target $t$ per unit time, we vary the compression level while remaining subject to the flow constraint (\textbf{line 6}). We store in the vector $R$ the compression level for each target. We store the loss in accuracy for $t$ subject to compression level $R(t)$, comparing the accuracy due to the best quality $Q((\mathcal{S}(t),*).a$ (\textbf{line 7}). The weights of the tree are updated considering the used bandwidth (\textbf{line 8}). Both the compression levels and the loss for each target are returned.

\subsection{Cost of a Coverage}

The cost of a connected coverage formation is parameterized by $\alpha$. This exogenous parameter weights the importance given to the the accuracy of the tasks. Notice that the importance given to task accuracy opposes to the minimization of the number of \Us employed.  Therefore the cost is a linear combination of the average loss in accuracy, and the percentage of used \Us to cover the new target in $\tau_{i}$ which was not present in the previous formation $\tau_{i-1}$, the cost is computed as $\textsc{Cost}_\alpha$:
\begin{equation}
    \alpha \cdot \frac{\sum_{t \in \mathcal{C}(\tau_i)} L(t)}{ |\mathcal{C}(\tau_i)|} + (1-\alpha) \cdot \frac{|V_i - V_{i-1}|-1}{|\mathcal{U}| - |V_\Psi \cup V_{i-1}|} 
\end{equation}

\begin{figure}[b]
\centering
\includegraphics[width=0.8\linewidth]{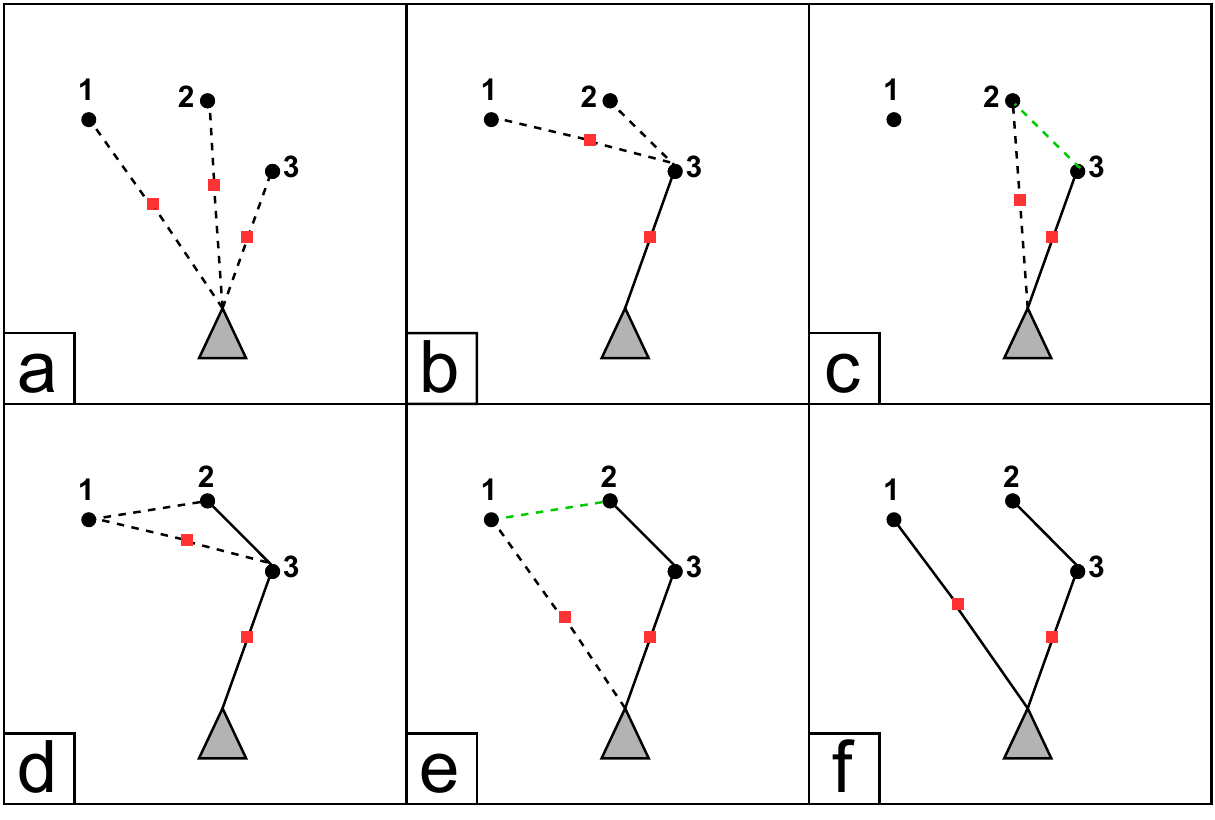}
\caption{\greedyprob algorithm example}
\label{fig:greedy_example}
\end{figure}

\subsection{\textsc{Greedy}-\gls{prob} Example Execution}
Figure \ref{fig:greedy_example} shows an example of execution of \greedyprob.
The gray triangle is the edge server $\sigma$. The black dots and red squares represent the target and relay positions, respectively. 
Figure \ref{fig:greedy_example}-a shows three temporary coverage $\tau_\text{temp}$, each covering a different target. The cost of each of the coverage is compared (algo. 1, \textbf{line 7}). 
Say $\tau_\text{temp}\langle \sigma, t_3\rangle$ is the cheapest coverage among them, that is the tree covering $\sigma$ and $t_3$. At the subsequent iteration shown in Figure \ref{fig:greedy_example}-b, two Triangular Steiner Trees covering $\sigma$, $t_3$ and a new target among the remaining uncovered ones in $\mathcal{T} - \mathcal{\widehat{T}}$ (i.e., $t_2$ and $t_1$) are proposed.
Say the tree $\tau_\text{temp}\langle \sigma, t_3, t_2 \rangle$ is the cheapest coverage among them. In Figure \ref{fig:greedy_example}-c the cost of $\tau_\text{temp}\langle \sigma, t_3, t_2 \rangle$ is compared with a line of sight coverage  $\tau_\text{los}\langle \sigma, t_2 \rangle$. The cheapest coverage among the two becomes the one to grow from the subsequent iterations (algo. 1, \textbf{line 14}).
Say the cheapest coverage among them is $\tau_\text{temp}\langle \sigma, t_3, t_2 \rangle$. In Figure \ref{fig:greedy_example}-d we see two grown versions of the tree covering $t_1$, whereas in Figure \ref{fig:greedy_example}-e we see a line of sight coverage of $t_1$. 
Say that comparing the cost of $\tau_\text{temp}\langle \sigma, t_3, t_2, t_1 \rangle$, and $\tau_\text{los}\langle \sigma,  t_1 \rangle$, the cheapest is the line-of-sight version.
The tree $\tau_\text{los}\langle \sigma,  t_1 \rangle$ becomes the new tree to grow from the subsequent iterations.
$\tau_\text{temp}\langle \sigma, t_3, t_2 \rangle$ is archived in $\Psi$. There are no more targets to cover. $\tau_\text{los}\langle \sigma,  t_1 \rangle$ is \mbox{ archived in $\Psi$ the algorithm stops returning $\Psi$}.

\subsection{Properties of \greedyprob}

\begin{lemma} 
\label{lem:cost}
Computing the compression level assignment has polynomial time complexity of $O(|\mathcal{U}|^2)$.
\end{lemma}

\begin{proof}[Proof sketch]
To measure the cost of a coverage tree $\tau_i$, the set of targets in the tree $\mathcal{C}(\tau_i)$ is sorted by their expected transmission load in ascending order. Sorting requires $O(|\mathcal{T}| \log |\mathcal{T}|)$ time complexity. 
The for loop iterates over the targets in $\tau_i$ first computing the bottleneck bandwidth for $t$, having approximately the cost of a Depth First Search and Shortest Path, that is $O(\log |V_i|)$ for the tree. Iterating over the compression levels to find the highest resolution to fit the bandwidth has constant complexity $|\mathcal{L}|$ i.e., the cardinality of the discrete set of possible compression levels. Iterating over the edges of the path $P$ to update the residual bandwidth has cost $O(\log |V_i|)$.
Other assignments have evident constant complexity.
By noticing that $|V_i| = O(|\mathcal{U}|)$ the overall time complexity of computing the cost of a coverage tree is $O(|\mathcal{T}| \log |\mathcal{U}|)$. The complexity further simplifies by considering $|\mathcal{T}| = O(|\mathcal{U}|)$, thus resulting in $O(|\mathcal{U}|^2)$. %{\color{red} further simplify? $O(|\mathcal{U}|^2)$ is my guess, if we say that T=O(U)}
\end{proof}

\begin{theorem}[Time Complexity of \greedyprob] 
\greedyprob with input $\mathcal{T}$ targets sets has polynomial time complexity of $O(|\mathcal{U}|^6)$.
\end{theorem}

\begin{proof}[Proof sketch]
The while loop  is executed, in the worst case, until all the targets in $\mathcal{T}$ are included in the final solution $\Psi$. %It may happen that the loop terminates earlier, that is when the number of devices is not sufficient to cover none of the remaining targets. %In that case, the pending local tree $\tau_\text{loc}$ is added to the final tree $\tau_\text{tot}$ without further expansion, and the algorithm terminates (\textbf{lines 9-11}).
Within the while loop, a for loop iterates over the set of uncovered targets. For each of them a Triangular Steiner Tree $t_\text{temp}$ is computed, and the time complexity is bounded by $O(|\mathcal{T}|^4)$ \cite{senel2011relay}. The cost of each tree is computed with complexity  $O(|\mathcal{U}|^2)$ as shown in Lemma \ref{lem:cost}.
Once the best candidate target to cover has been chosen, the Triangular Steiner Tree $t_\text{los}$ of the shortest path path towards the target, and the relative cost $c_\text{los}$ are computed. The time complexity to find a stripe can be considered constant in time. The overall time complexity of \greedyprob is thus given by: $O(|\mathcal{T}|(|\mathcal{T}|(|\mathcal{T}|^4 + |\mathcal{U}|^2) + |\mathcal{U}|^2)) = O(|\mathcal{U}|^6)$.

\end{proof}

\section{Performance Evaluation}\label{sec:perf_section}

We extensively evaluate {\FW} through simulation (Sect.  \ref{sec:simulations}) as well as real-world experiments (Sect. \ref{sec:exp_section}).

\subsection{Evaluation Setup}\label{sec:setup}

\textbf{Application.}~We consider a monitoring application where \Us need to perform image classification or object detection tasks on \textit{target} locations by sampling images at given frame rate (e.g., 24 \gls{fps}). We adopt (i) \textit{ResNet-50}, a \gls{cnn} with 50 layers \cite{he2016deep}; (ii) \textit{ResNet-152}, an extended version with 152 layers \cite{he2016deep}; (iii) \textit{DenseNet} \cite{huang2017densely}, which consists of a Dense Convolutional Network (i.e., each layer is connected to all the other layers in a feed-forward fashion); (iv) \textit{MobileNet-V2} \cite{mobilenetv2}, a new neural architecture for mobile devices; (v) YoloV4, the state-of-the-art model for object detection. All the models were trained on the ImageNet database \cite{deng2009imagenet}.

\textbf{Scenarios}. 
To emulate common scenarios for UAVs, we use the five scenarios described in Section \ref{sec:qodl}, i.e., \textit{
Maritime, Search-and-Rescue, Wildlife, Tools, Pets}. We also design an \textit{Urban} reconnaissance scenario including various objects, such as \textit{wreck, fireboat, ambulance, police van, revolver, crate, packet, backpack, mountain bike, motor scooter}. To ensure repeatability of our experiments, we let the UAVs sample images from a labeled subset of ImageNet. Where not otherwise stated, each target location generates 500 tasks (images) uniformly sampled among these classes. 

\textbf{Metrics}. 
We measure the \textit{Percentage of Accomplished Tasks}, defined as the ratio between the number of successfully completed tasks (according to Definition \ref{def:task}) and the number of the generated tasks. The accomplishment of a task is influenced by its  deadline $\Delta$. In order to study the performance of \gls{prob} at varying application scenarios, we let $\Delta$ vary: low values represent delay critical applications (e.g., intrusion detection), whereas high values, delay tolerant ones (e.g., agriculture).
%
%a task is completed if offloaded before its deadline $\Delta$, and the DL model at the edge server outputs a correct prediction for it.
%correctly classified at the edge by the \gls{dl} model. 
%
We also measure \textit{Computational Time}, that is the time required by the algorithms to output a connected coverage formation and compression levels for the targets.
%%
%To emulate different applications, we vary the value of $\Delta$, tightening the available time to offload and analyze a task. For example,  intrusion detection may require $\Delta = 0.1$ seconds to promptly take countermeasures, while food monitoring can have a more relaxed deadline, e.g.,  $\Delta=2$ seconds. 
%\fr{va bene, ma anche l'accuracy dobbiamo tenere in conto.} \ac{see above}

%Opt-\gls{prob}
%\greedyprob
\textbf{Comparison.}~We evaluate \FW through real-field experiments and simulation, considering both the optimal solution \optprob and the greedy algorithm \greedyprob, against \gls{stba} \cite{nguyen2019new}. \gls{stba} is a state-of-the-art networking-based approach that is the closest to our work. \gls{stba} covers a set of targets while providing network connectivity to the edge server. To find a connected tree, \gls{stba} uses a node-weighted Steiner tree algorithm, which computes a set of Fermat points to place relays, and then computes a tree among the targets and the edge-server, minimizing the needed \Us. To allow for a fair comparison, we enhance \gls{stba} with data compression in three variants: 1) H-\gls{stba}, which does not compress data, but uses the \textbf{H}ighest available quality for collected data ($l=1$); 2) M-\gls{stba} which uses the \textbf{M}edium compression ($l = 50$); and 3) L-\gls{stba} which uses an extreme compression ($l=100)$ resulting in the \textbf{L}owest data quality.

\subsection{Simulation Results}\label{sec:simulations}

We used the NS-3 network simulator \cite{nsnam}, setting most of parameters in line with the devices used in our real-field experiments (e.g., WiFi interface 802.11n at 2.4 GHz), and testbed measured values (UAVs transmission range is $16m$, sensing radius $1m$, and maximum speed $5m/s$)\footnote{The code is available at https://github.com/flaat/AA-UAV}. The simulated area is a square of $500 \times 500 m$, with an edge-server positioned in the center of the bottom border. The number of targets varies from 4 to 50, and the number of UAVs from 4 to 50.  %If not specified otherwise, the plots are made considering 5 targets randomly distributed in the area and 10 \Ms.

\begin{figure}[h] 
  %\centering
  \begin{minipage}{0.48\linewidth}
  \includegraphics[width=1\textwidth]{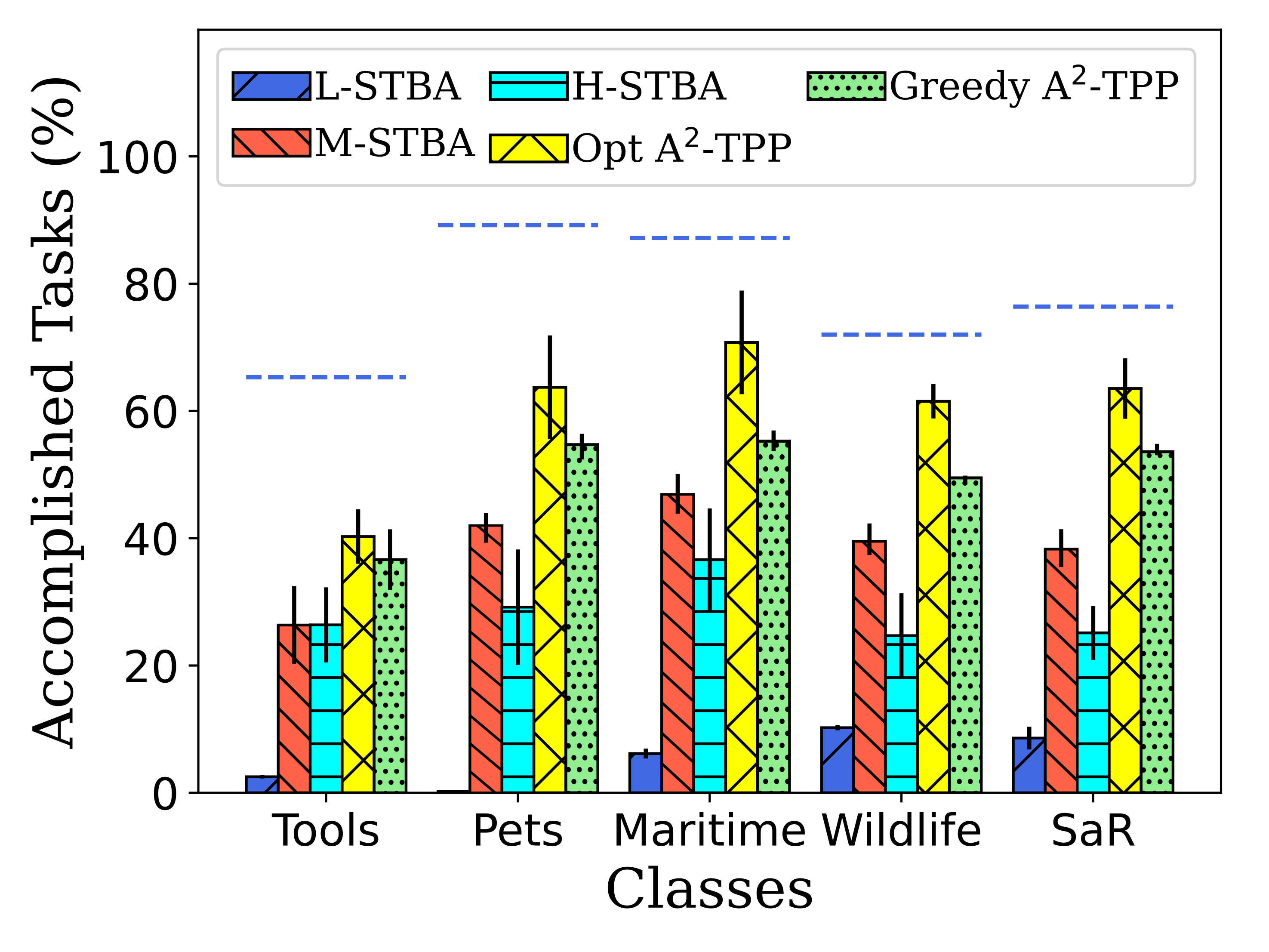}
  \caption{{\small Accomplished tasks  (\%), $\Delta=0.1$sec}}
    \label{fig:multi_dense_all}
  \end{minipage}  
  \hfill
  %\centering
    \begin{minipage}{0.48\linewidth}
  \includegraphics[width=1\textwidth]{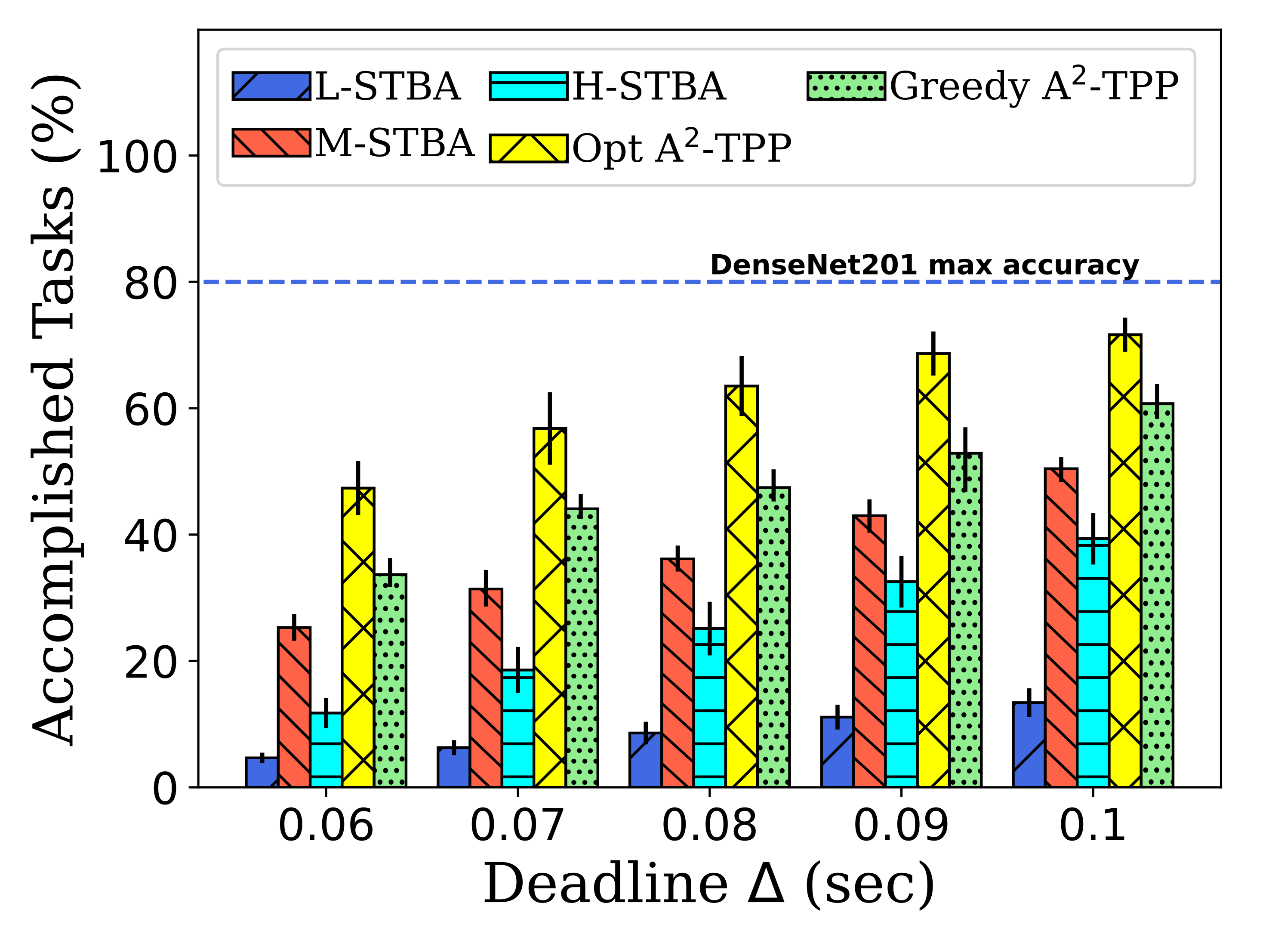}
  \caption{{\small Accomplished tasks (\%) at increasing of $\Delta$}}
\label{fig:multi_dense_sar}
  \end{minipage}
\hfill
  %\centering
  \begin{minipage}{0.48\linewidth}
  \includegraphics[width=1\textwidth]{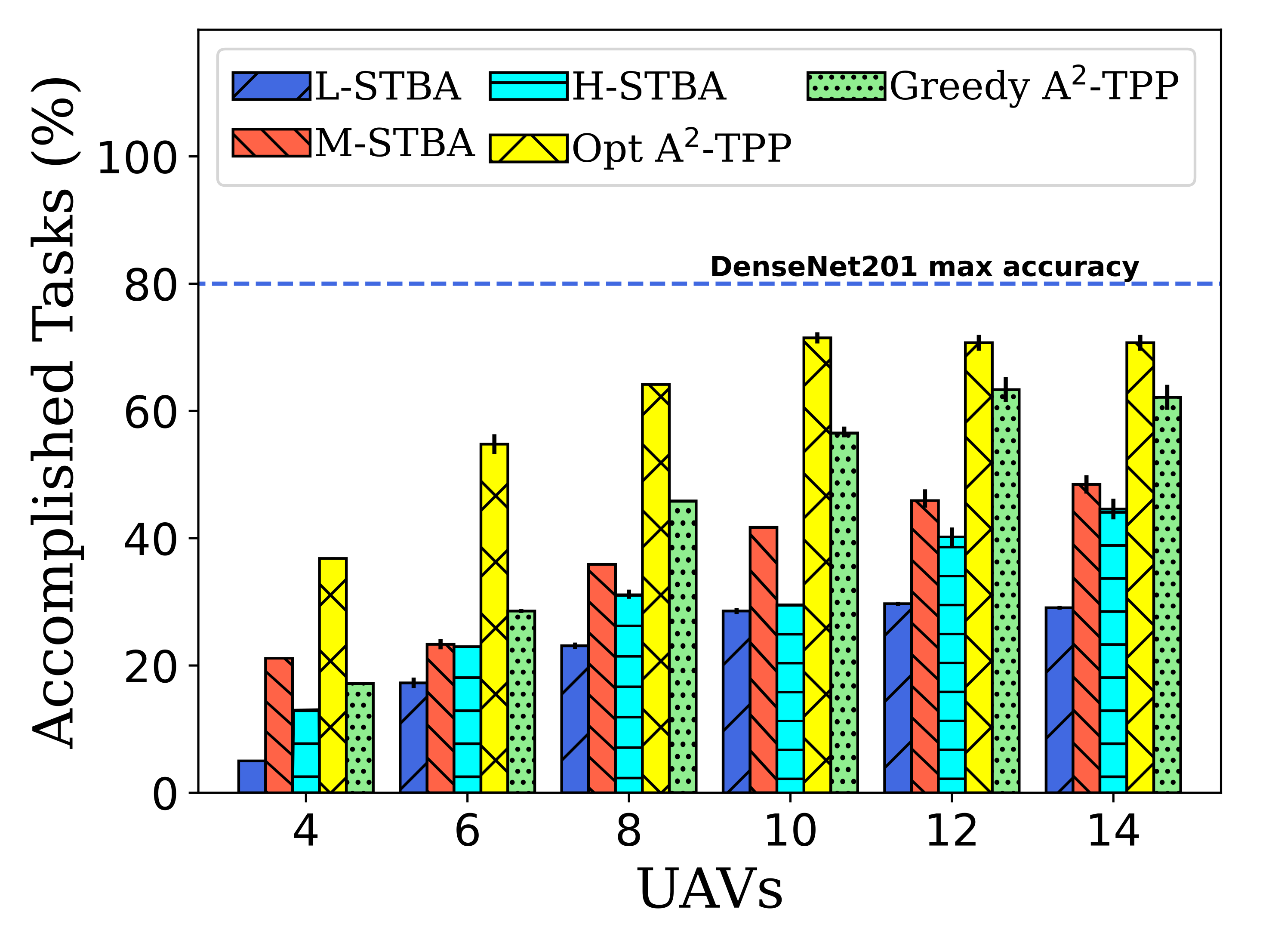}
  \caption{{\small Accomplished tasks (\%) with 6 Targets, $\Delta=0.1$sec}}
    \label{fig:nvar_drones_6_targets}
  \end{minipage}  
      \hfill
    %\centering
    \begin{minipage}{0.48\linewidth}
  \includegraphics[width=1\textwidth]{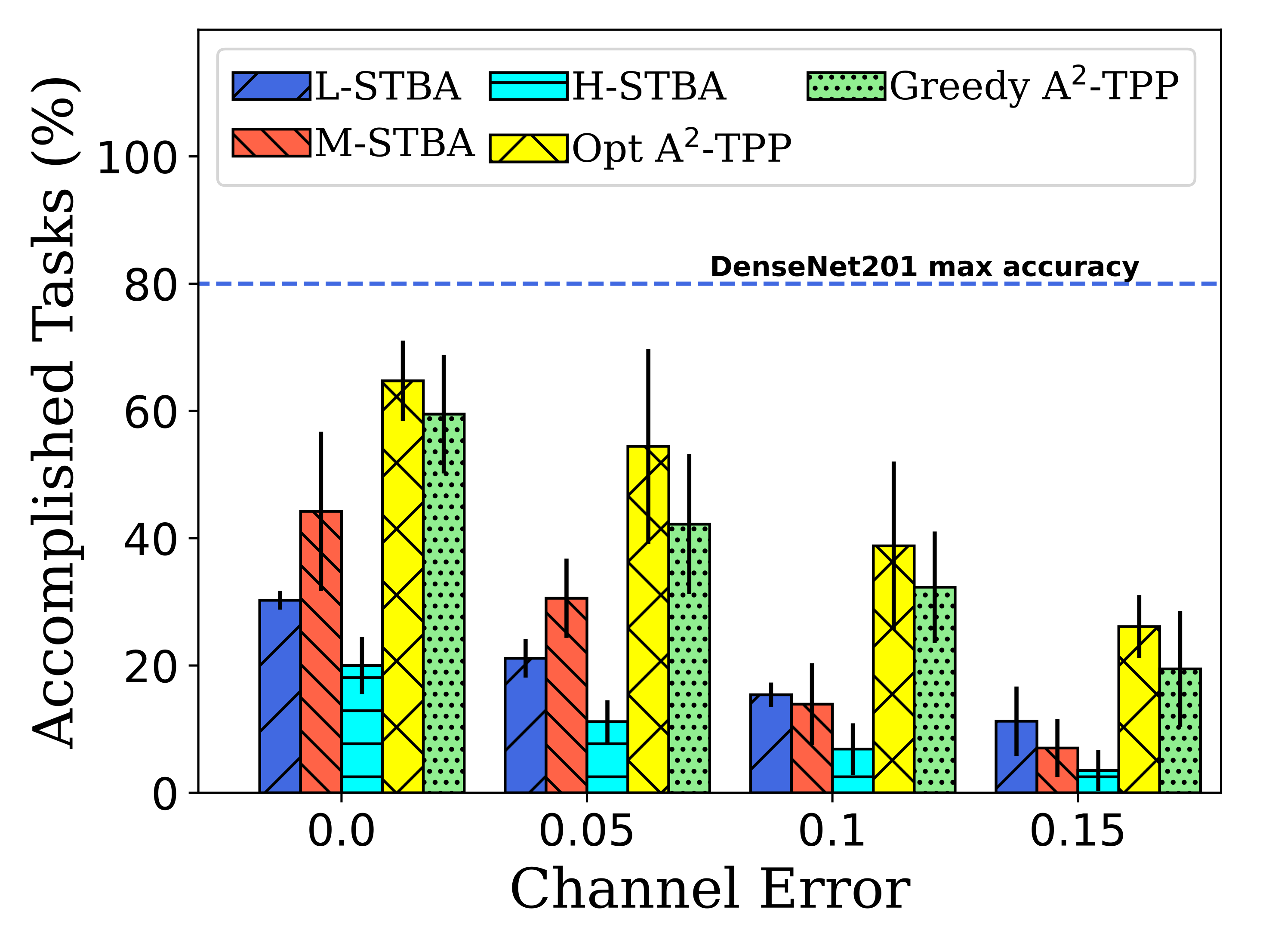}
  \caption{{\small Accomplished tasks (\%), 4 targets, $\Delta=0.1$sec}}
\label{fig:channel_error}
  \end{minipage}
%  \end{figure}[t] 
 \hfill
  %\centering
    \begin{minipage}{0.48\linewidth}
  \includegraphics[width=1\textwidth]{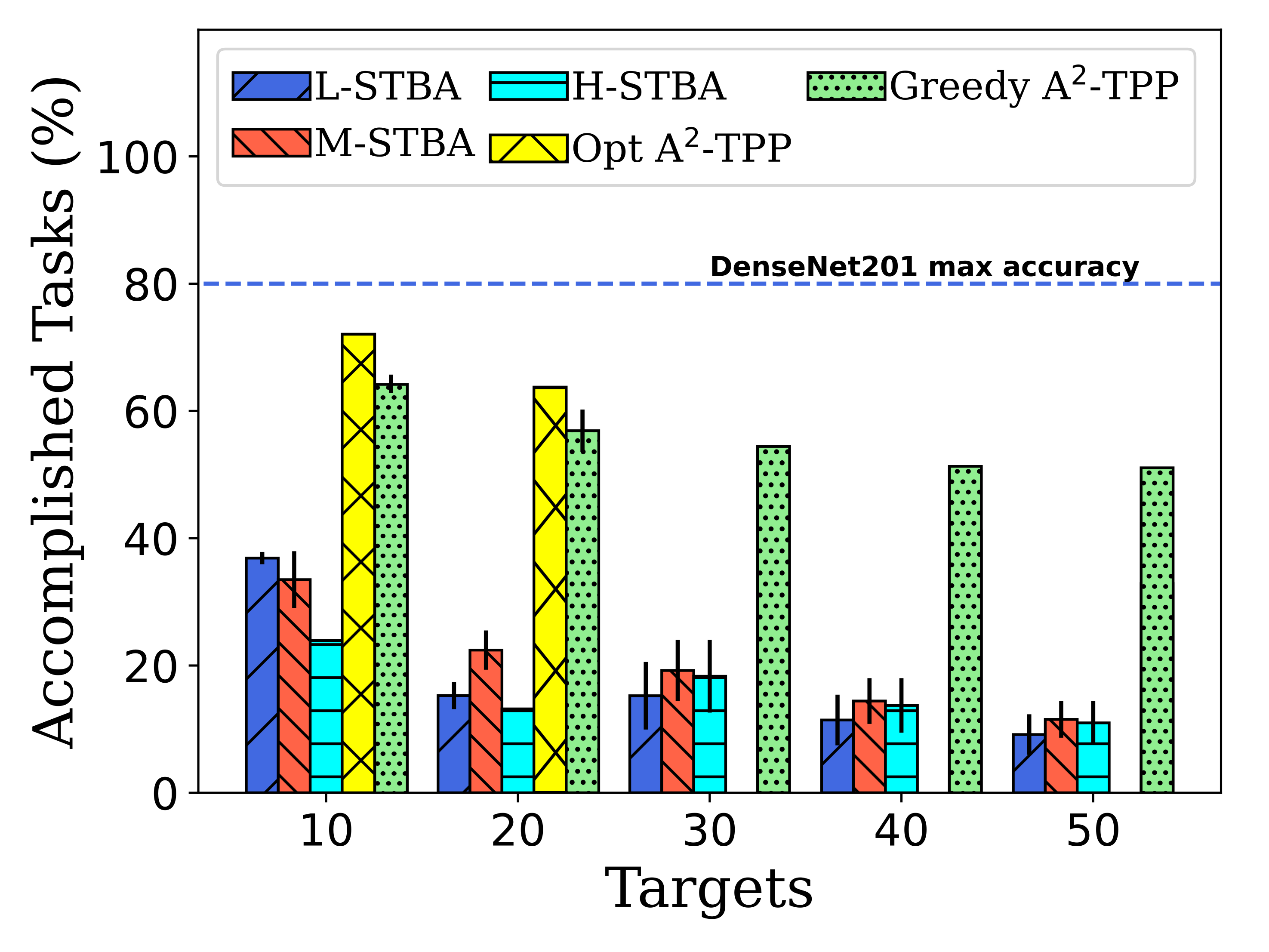}
  \caption{{\small Accomplished tasks (\%), increasing targets}}
\label{fig:nvar_targets}
  \end{minipage}
    \hfill
    %\centering
    \begin{minipage}{0.48\linewidth}
  \includegraphics[width=1\textwidth]{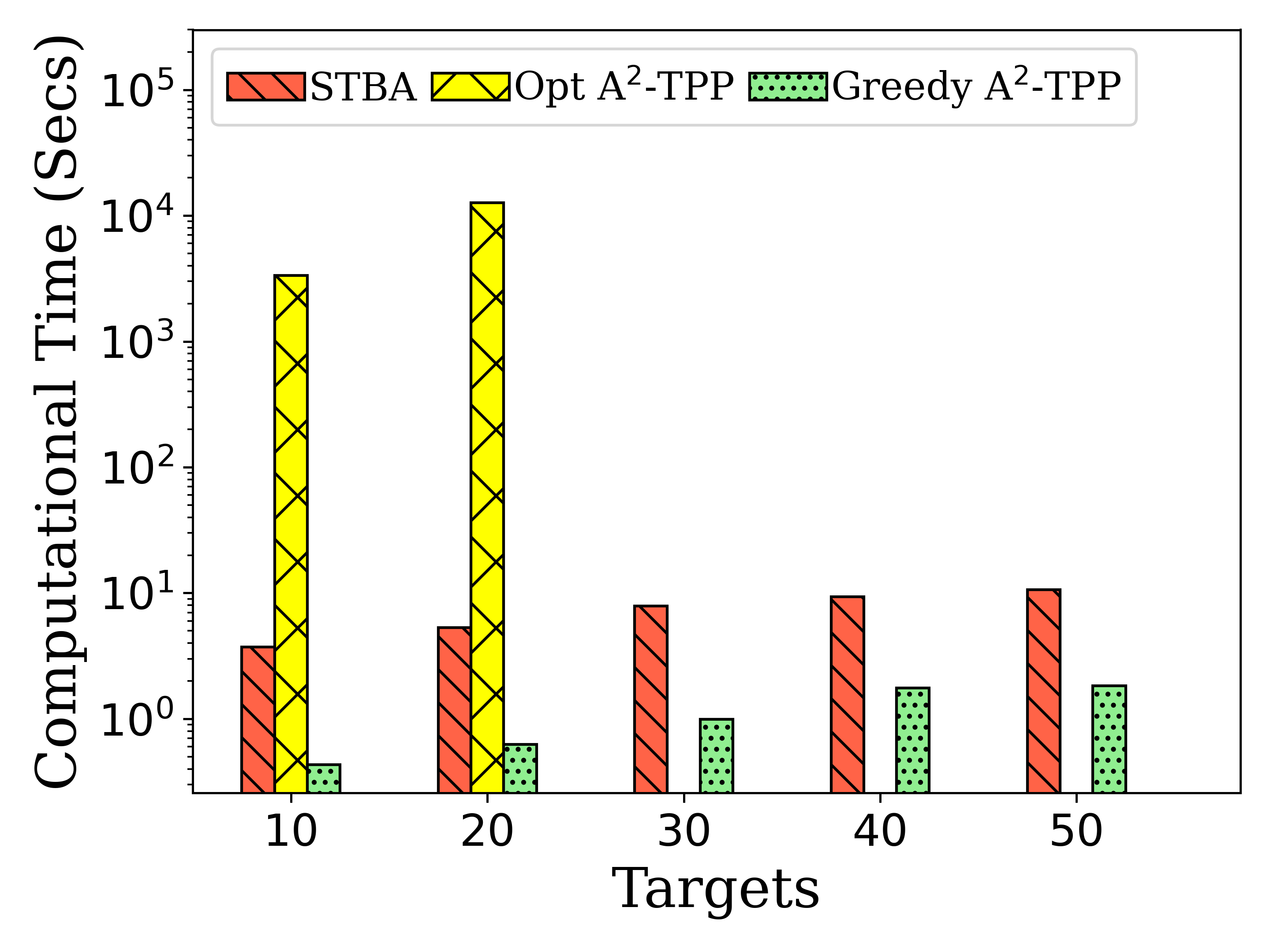}
  \caption{{\small Computational time (sec)}}
\label{fig:times}
  \end{minipage}
\end{figure}  

\vspace{0.1em}
\subsubsection{Multiple Scenarios}
Figures \ref{fig:multi_dense_all} illustrates the efficacy of our solution for different scenarios, reporting also the theoretical upper bound (blue dotted line). 
In the most challenging scenario, i.e., \textit{Tools}, with a strict task deadline ($\Delta=0.1$sec) and DenseNet-201 \gls{dl} model, \textbf{\optprob completes $41\%$ of tasks, with an improvement of $52\%$ respect the best \gls{stba} variant}, i.e., H-\gls{stba}, which completes less than $27\%$ of tasks. The theoretical upper-bound for DenseNet-201 in the same scenario is $60\%$, meaning that under ideal network conditions of zero latency and no compression, the \gls{dl} model would correctly classify only $60\%$ of the tasks (Figure \ref{fig:qol_levels}.a show the complexity of predicting tools images, even for un-compressed images). 
%: low contrast and the small size of the objects disrupt the application which achieves at most $60\%$ of accuracy with DenseNet-201 as model.  
In the case of \textit{Pets and Maritime}, 
\optprob reaches the highest percentage of accomplished tasks  --- $65\%$ and $70\%$ respectively --- where the upper-bounds are $90\%$ and $88\%$. \textbf{The improvement with respect to the best STBA variant, i.e., M-\gls{stba}, is $55\%$ and $50\%$}. \textit{Pets} require a compression level lower than $l=50$ (see Figure \ref{fig:qol_levels}.a) to achieve satisfactory performance, forcing both \optprob and \greedyprob to select a medium compression level, more similarly to M-\gls{stba}. %This trend is even more clear with MobileNet-V2 in Figure \ref{fig:multi_mobile_all}. The model requires high-resolution images to achieve satisfactory performance, while low-resolution images used by L-\gls{stba} leads to few accomplished tasks. With MobileNet-V2, \FW executes $10-15\%$ more tasks than M-\gls{stba}, with around $64\%$ and $70\%$ of accomplished tasks for \textit{Pets and Maritime}. 
\\
In \textit{Wildlife} and \textit{Search-and-Rescue} (SaR), the gap between both the \gls{prob} versions and STBA variants increases significantly. \optprob and \greedyprob complete respectively $60\%$, $63\%$ and $49\%$ and $54\%$ of tasks,  against $39\%$ and $37\%$ of M-\gls{stba}. The motivation behind this sharp improvement is the use of the \texttt{A$^2$-TA}, which understands that even high compressed images can achieve satisfactory performance. Therefore, both our solutions can achieve high accuracy with low network usage, executing the tasks within their deadline $\Delta=0.1$ seconds. \greedyprob completes $20\%$ and $31\%$ more tasks than M-\gls{stba}. %In this scenario also L-\gls{stba} achieves better performance, thanks to the accuracy of low quality images. 
%Figure \ref{fig:multi_mobile_all} confirms a similar trend with MobileNet-V2 model: \FW accomplishes $40\%$ and $30\%$ more tasks with respect to M-\gls{stba} for \textit{Wild} and \textit{SaR} scenarios, respectively. A general decrease of overall performance is due to the lower accomplishment of MobileNet-V2 model (see Figure \ref{fig:qol_levels}.a).
%, while L-\gls{stba} drops to a few accomplished tasks. 
%always remain way better outperforming from $10\%$ up to $26\%$ the \gls{stba} solutions.

\vspace{0.1em}
\subsubsection{Urban Scenario}

Figure \ref{fig:multi_dense_sar} shows the performance in the \textit{Urban} scenario as a function of task deadline $\Delta \in \{ 0.06, 0.07, 0.08, 0.09, 0.1 \}$, when DenseNet-201 is employed. 
%In the case of MobileNet-V2 (Figure \ref{fig:multi_mobile_sar}), increasing the deadline $\Delta$ almost all the solutions accomplish more tasks, as they have more time to offload and process tasks. In particular, H-\gls{stba} has the best improvement, as it uses the highest resolution images, which require more time to traverse the network, while L-\gls{stba} has no benefits because its low resolution images are easily offloaded under the deadline, but often misclassified by \gls{dl} model. %Instead, M-\gls{stba} and \FW have a moderate improvement as they both use moderate compression levels.
%
\optprob accomplishes tasks up to $72\%$ in the case of  $\Delta=0.1$sec, while the best variant M-\gls{stba} achieves only $48\%$ of tasks at the same $\Delta$. \textbf{\optprob accomplishes $58\%$ more tasks than M-\gls{stba} with the tightest deadline}, as it adapts the compression of images to meet the latency constraint. The plot also confirms the performance of \optprob that outperforms the network-based approaches (i.e., M-\gls{stba}) up to $45-50\%$. \textbf{\greedyprob follows the \optprob trend always remaining widely above the performance of \gls{stba} solutions.} Figure \ref{fig:nvar_drones_6_targets} shows the percentage of accomplished tasks as function of the number of UAVs, with 6 targets randomly distributed in the area. We employ DenseNet-201, which achieves a maximum accuracy of $80\%$, and set $\Delta=0.1$sec. Both \optprob and \greedyprob outperform the STBA variants, as they cover all the targets with only 8 UAVs. Conversely, the \gls{stba} variants require at least 10 UAVs to cover all the targets, and achieve lower performance. \optprob covers $15-20\%$ more targets than \gls{stba} algorithms, in all the scenarios, completing $69\%$ of tasks (using 10 UAVs), while the best variant M-\gls{stba} accomplishes only $42\%$ of tasks with the same number of UAVs. We can notice how \greedyprob performs better as quickly as number of UAVs grow reaching similar performance of \optprob.

% So far we have considered an ideal channel to focus on the improvement achieved by our system. We now further investigate the efficacy of our approach by evaluating the robustness to channel errors.

\vspace{0.1em}
\subsubsection{Robustness to Channel Errors}

In Figure \ref{fig:channel_error} we plot the percentage of accomplished tasks by varying the probability of channel error $\psi \in \{0, 0.05, 0.1, 0.15\}$, in a setting with 20 UAVs and 4 targets. Both \optprob and \greedyprob are the most robust algorithms, increasing their improvement with respect to \gls{stba} variants. \textbf{\optprob completes  up to 170\% more tasks than the other approaches}.
On the other hand, M-\gls{stba} and H-\gls{stba} experience severe delays and drastic performance reduction due to frequent TCP re-transmissions, which introduces additional data in the network, further overloading communication links.

\vspace{0.1em}
\subsubsection{Scalability} Figure \ref{fig:nvar_targets} investigates the percentage of accomplished tasks in a scenario with 50 UAVs, varying the number of targets from 10 to 50. We do not include the \optprob when the targets are more than 20, due to prohibitive computational time. This result underlies the huge benefit introduced by the polynomial time solution \greedyprob, which scales gracefully when the problem instance grows in complexity. The figure shows that \greedyprob has near optimal performance with 10 targets, accomplishing $63\%$ of the tasks, while L-\gls{stba} accomplishes only $38\%$ of them. %Furthermore, when the number of targets increases, \optprob accomplishes around $63\%$ of tasks with 20 targets, while M-\gls{stba} completes only $22\%$ of tasks. 
All the algorithms have a slightly decreasing trend as the number of targets increases, as the UAVs have to offload more tasks with possible network congestion and missed deadlines. The \gls{stba} variants quickly drop their performance due to congestion and long delays, while \greedyprob is able to keep satisfactory performance around $50\%$, trading off compression and accuracy to cover all targets and offload their data. With 50 targets \greedyprob accomplished 5 times the tasks of the best \gls{stba} variant.
Finally, in Figure \ref{fig:times} we investigate the computational time. We restrict the time to a maximum of 5 hours (18000 seconds), and we consider no solutions after that time. We consider a general \gls{stba} instance without compression levels, as they do not affect the execution time. While \optprob has very huge computational times even with 10 targets, \textbf{\greedyprob is 15x faster than the \gls{stba} solutions}. 

\subsection{Experimental Testbed Results}\label{sec:exp_section}

We now evaluate the performance through our experimental testbed. \textbf{The testbed is composed of 4 UAVs and an edge server with dedicated GPU}. We emulate missions with up to 4 targets. We run 10 experiments for each scenario and we average the results. Each UAV includes a DJI Mavic Air 2 drone, mounting a Raspberry PI 4 model B and a power bank, as shown in Figure \ref{fig:testbed}. The on-board Raspberry PI, powered by the power bank is used to generate and pre-process tasks, and to offload them to the edge according to the optimization plan. For repeatability and to emulate different scenarios, we sample images from the ImageNet dataset \cite{russakovsky2015imagenet}.

\begin{figure}[h]
    \centering
    \begin{minipage}{0.45\linewidth} 
    \resizebox{\hsize}{!}{
    \begin{tabular}{ll}
        \hline
Field                            & *****                  \\ \hline
        \textbf{Time}              & 9:00-18:30 a.m.        \\
        \textbf{Temperature}       & +4 - 15°C              \\
        \textbf{Wind Speed}              & 0.0 to 4.3 m/s         \\
        \textbf{Field Size}              & 65x35(meters)          \\
        \textbf{Nr. Of UAVs}        & 4                      \\
        \textbf{Nr. Of Targets}       & Variable (from 1 to 4) \\
        \textbf{Humidity}                & 60\% - 77\% \\
        \textbf{AMSL} & 2 meters               \\ \hline
    \end{tabular}}
    \captionof{table}{\small Experimental Setting}\label{tab:exp_cond}
    \end{minipage}
    \hfill
\begin{minipage}{0.48\linewidth} 
%\includegraphics[width=\linewidth]{img/realFieldExp/schema/new_schema/Star config revisited.png}
%\caption{Scenario}
%\label{fig:realscenario}
\includegraphics[width=\linewidth]{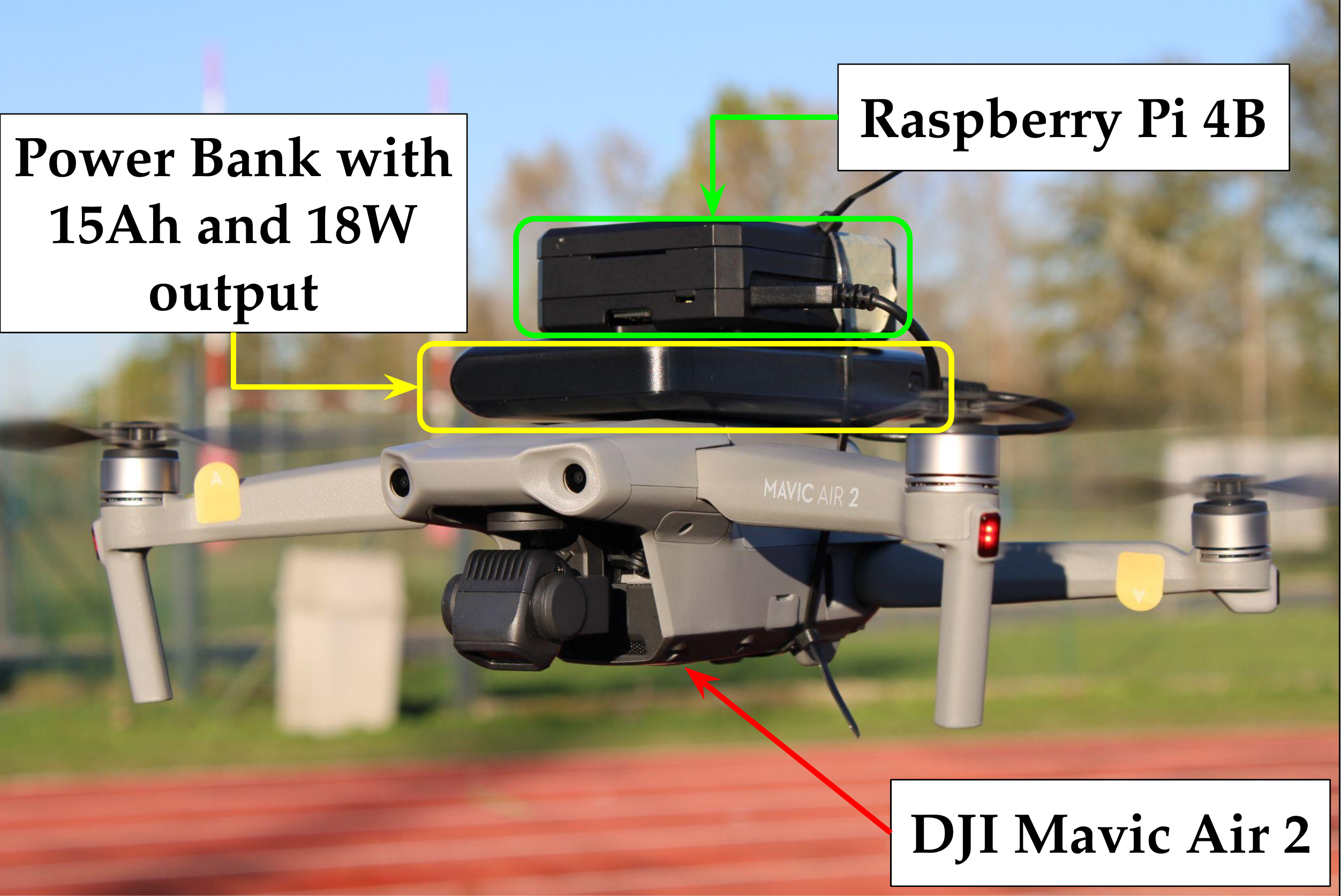}
% https://docs.google.com/drawings/d/1dMD4PO3muUcKnqV1AZ9GoL8efdE50CG0twXobWBikPQ/edit
\caption{\small UAV implementation.}
\label{fig:testbed}
\end{minipage}
\end{figure} 

The edge server is a Jetson Nano board, used to run the \gls{dl} models and execute tasks. It mounts a Raspberry PI for computation and communication. TCP links are established for reliable connectivity. Considering the limited capabilities of the edge server, we execute only ResNet-50 and MobileNet-V2 on the Jetson Nano, which have approximately 0.03 seconds of inference time \cite{nvidiainference2021}. 
For DensNet-201, ResNet-152, and YoloV4, we used a laptop with an NVIDIA RTX-2060 \gls{gpu}. In the experiments, we consider up to 4 targets placed at around 15 meters from the edge. Table \ref{tab:exp_cond} reports the experimental settings in the \textit{Urban} scenario. % Figure \ref{fig:realscenario} depicts our testbed scenario, where red triangles represent targets locations, while the black and white circle indicates the edge-server location. %$n$ targets, we refer to the targets $\{1, ..., n\}$. For example, a mission with 2 targets considers only the targets numbered \textit{1} and \textit{2}. 
%\begin{figure*} 
%\centering
  %\begin{minipage}{0.24\textwidth} 
  %\includegraphics[width=\textwidth]{img/realFieldExp/4_01_22/targets_vs_delay_c%onfig_star.png} 
  %\caption{Avg. Image Delay (Star)}
  %  \label{fig:star_delay_exp}
  %\end{minipage}  
  %\hfill
  %\centering
 % \begin{minipage}{0.24\textwidth}
%  \includegraphics[width=\textwidth]{img/realFieldExp/4_01_22/goodput.png}
%  \caption{Goodput at increasing of Nr. Targets (Star)}
%\label{fig:goodput_star}
 % \end{minipage}
%  \hfill
%ing-based metrics such as delay and throughput, which ultimately causes over-provisioning. 
%
%
\begin{figure}[t] 
  \centering
  \begin{minipage}{0.48\linewidth}
  \includegraphics[width=\textwidth]{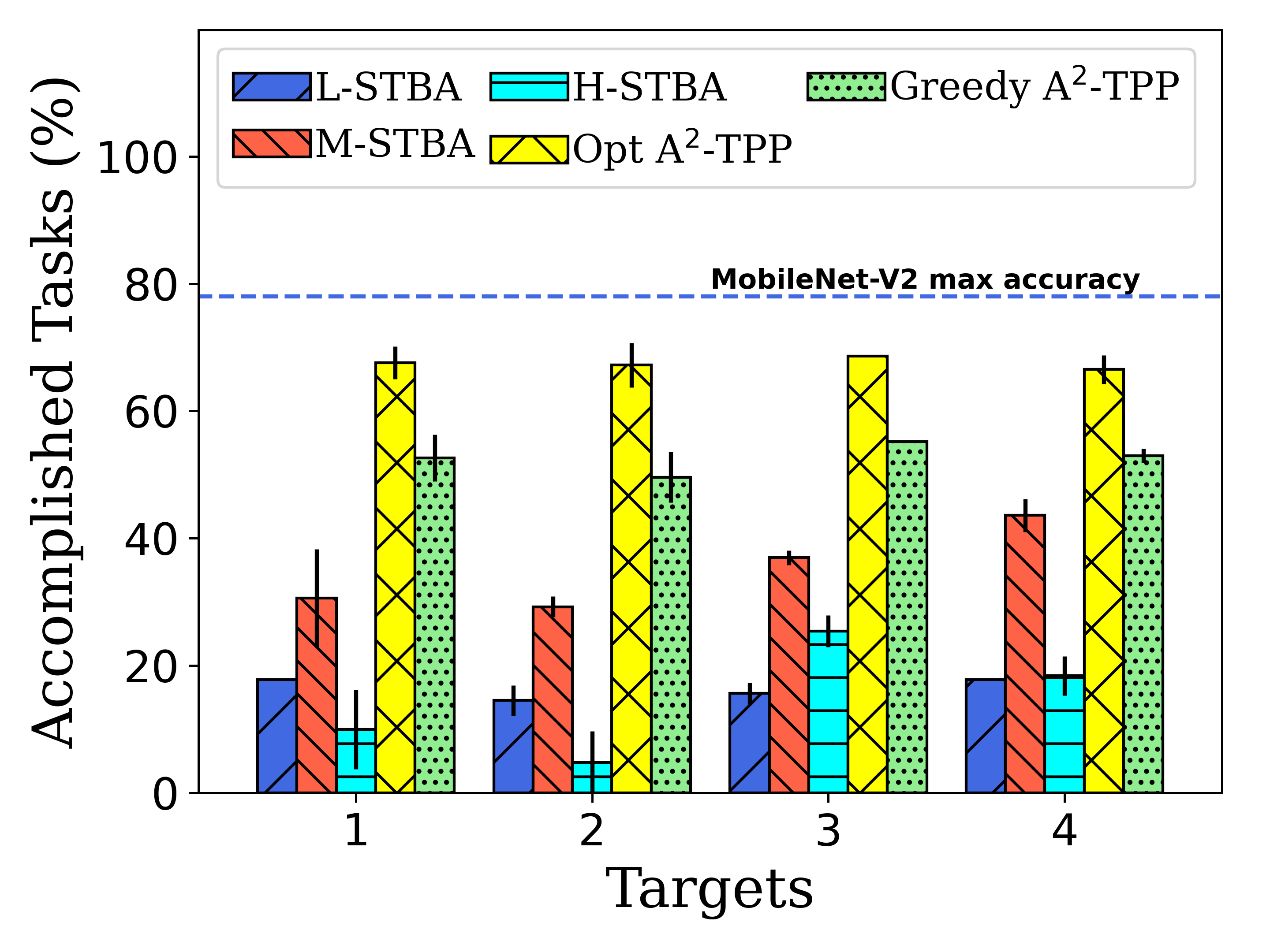}
  \caption{\small Accomplished Tasks (\%) at increasing targets, MobileNet-V2, $\Delta=0.4$sec}
    \label{fig:star_real_exp_sources_task_vs_targets}
  \end{minipage}  
  \hfill
  \centering
  \begin{minipage}{0.48\linewidth}
  \includegraphics[width=\textwidth]{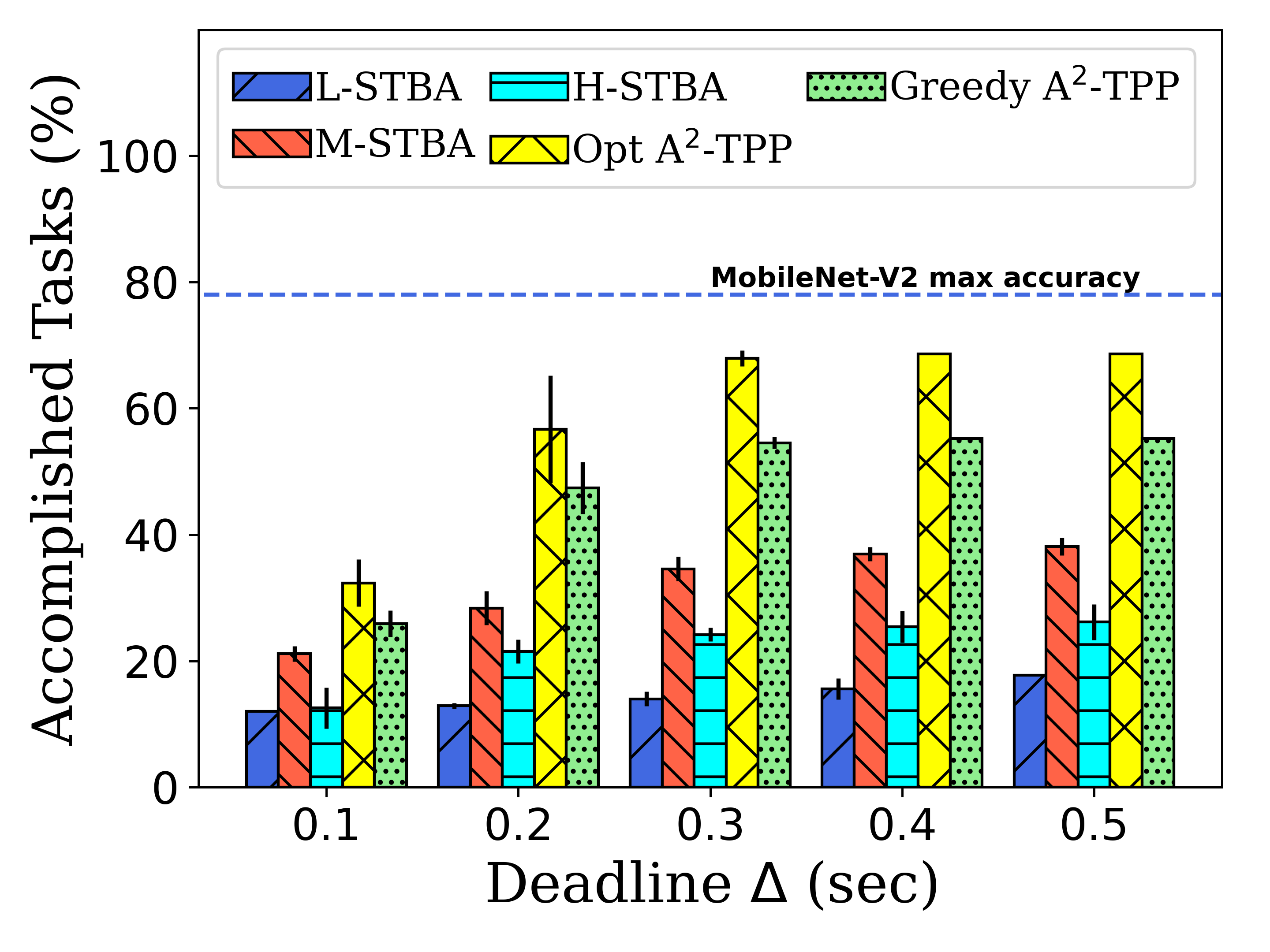}
  \caption{\small{Accomplished Tasks (\%) at increasing of $\Delta$, using 3 targets and MobileNet-V2}}
\label{fig:perc_task_increasing_delta_star_exp}
  \end{minipage}
%  \centering
%    \begin{minipage}{0.325\textwidth}
 % \includegraphics[width=\textwidth]{img/Plots_28_01_2022/plot_bat_task_vs_delay.png}
  %\caption{Required Delays to Accomplished Tasks  $\%$, MobileNet-V2}
%\label{fig:required_delays_for_task_accomp}
 % \end{minipage}
%  \vspace{-0.1in}
\end{figure}  
The first set of experiments evaluates the impact of increasing the number of targets (from 1 to 4), with MobileNet-V2.  
Figure \ref{fig:star_real_exp_sources_task_vs_targets} shows the percentage of accomplished tasks at the edge-server with a task deadline of $\Delta = 0.4$sec. %We recall that a task is successfully executed if it is offloaded within the deadline $\Delta$ and correctly classified by \gls{dl} model. 
The plot shows that the \optprob finds the best trade-off between accuracy and data compression. It completes more than 67\% of the tasks, independently of the number of targets. This is close to the maximum performance achievable with the \gls{dl} model (i.e., 78$\%$), represented by the blue horizontal line. \textbf{\greedyprob instead reaches up to 57\% accomplished task, with a 20\% average improvement over the best STBA version} (M-STBA).
Conversely, the best \gls{stba} variant (i.e., M-\gls{stba}) does not complete more than 46\% of tasks, independently of the number of targets. In particular, with 2 targets, all \gls{stba} variants perform very poorly, completing less than 30\% of tasks. 
The superiority of \FW in both the approaches (Opt and Greedy) is confirmed by results on the percentage of accomplished tasks by varying the deadline $\Delta \in [0.1, 0.5]$ (see Figure \ref{fig:perc_task_increasing_delta_star_exp}). \texttt{Opt-}\gls{prob}
%almost doubles 
reaches an improvement over the percentage of executed tasks with respect to  M-\gls{stba} up to 76\% when $\Delta=0.3$sec.
The \greedyprob approach instead improves M-\gls{stba} results (when $\Delta$=0.3sec) around 50\% upholding our intuition.
%\\
We investigated the performance of \greedyprob and \optprob also when other \gls{dl} models are applied. Table \ref{tab:models_exp} summarizes the results in the case of $\Delta=0.3$sec and 4 targets, for ResNet50, MobileNet-V2 (executed on the Jetson Nano) and ResNet152, DenseNet201 and YoloV4 (executed on a laptop with a dedicated GPU). The results show that both our solutions outperforms all \gls{stba} variants independently of the applied model. In particular, with DenseNet201 \optprob has the best performance. 
\begin{table}[h]
%\centering
\resizebox{\columnwidth}{!}{%
\begin{tabular}{|c|c|c|c|c|c|}
\hline
\gls{dl} Model          & \optprob  & \greedyprob & L-\gls{stba} & M-\gls{stba} & H-\gls{stba} \\ \hline
ResNet50     & \textbf{66.64} & 62.88 &25.42  & 36.14  & 21.86   \\ \hline
ResNet152    & \textbf{67.89} & 65.27 &29.93  & 37.96  & 24.70   \\ \hline
DenseNet201  & \textbf{70.17} & 68.33 &32.92  & 41.87  & 26.99   \\ \hline
MobileNet-v2 & \textbf{69.29} & 57.36 &18.37  & 38.94  & 21.91   \\ \hline
YoloV4         & \textbf{59.32} & 51.3  &15.22  & 31.52  & 17.18   \\ \hline
\end{tabular}
}
\caption{Percentage of Completed Tasks, $\Delta=0.3$sec}\label{tab:models_exp}
\end{table}

\section{Conclusions}\label{sec:conclusion}
In this paper we proposed {\FW}, a novel \textit{application-aware} optimization framework for reliable and effective \gls{dl} task offloading in multi-hop UAVs networks. 
%
%show that the key innovation of \FW is to
For the first time, we considered the \textit{accuracy} and \textit{delay} requirements of the specific \gls{dnn} tasks, to jointly optimize task assignment and offloading.
Through extensive simulation, we demonstrated that \FW is able to deal with different network conditions, maximizing the application performance at the edge. \FW outperforms existing approaches, getting and average improvement w.r.t. the state-of-the-art algorithm of 38\%.
We finally validated our solution through real-field experiments, considering four DJI Mavic Air 2 \Us and a Jetson Nano board as edge server.  We share datasets and code with the research community to allow reproducibility.

% double blind, we can't put ack now
\newpage
\section{Acknowledgement of Support and Disclaimer} 
This work is funded in part by the G5828 "SeaSec: DroNets for Maritime Border and Port Security" project under the NATO’s Science for Peace Programme, by the National Science Foundation (NSF) grant CNS-2134973 and CNS-2120447, as well as by an effort sponsored by the U.S. Government under Other Transaction number FA8750-21-9-9000 between SOSSEC, Inc. and the Government. The U.S. Government is  authorized to reproduce and distribute reprints for Governmental purposes notwithstanding any copyright  notation thereon. The views and conclusions contained herein are those of the authors and should not be interpreted as necessarily representing the official policies or endorsements, either expressed or implied, of the Air Force Research Laboratory, \mbox{the U.S. Government, or SOSSEC, Inc.} 

\balance
\small
\bibliographystyle{IEEEtran}
\bibliography{IEEEabrv,biblio}

\end{document}